\documentclass[12pt]{article}
\usepackage[top=20truemm,bottom=20truemm,left=30truemm,right=30truemm]{geometry}
\usepackage{enumerate,indentfirst,braket,mathrsfs,mathtools,amsmath,amssymb,comment,bm,amsthm,amsfonts}
\usepackage{float}
\usepackage{wrapfig}
\usepackage{ascmac, fancybox}
\usepackage{bm}
\usepackage[toc, page]{appendix}

\theoremstyle{plain}
\newtheorem{thm}{Theorem}[section]
\newtheorem*{thm*}{Theorem}

\newtheorem{lem}[thm]{Lemma}
\newtheorem{prop}[thm]{Proposition}

\theoremstyle{definition}

\theoremstyle{remark}
\newtheorem*{rem}{Remark}
\newtheorem*{con*}{Conjecture}
\newtheorem*{example*}{Example}

\usepackage{ulem}

\title{A coding theoretic study of homogeneous Markovian predictive games}
\author{Takara Nomura and Akio Fujiwara%
\thanks{fujiwara@math.sci.osaka-u.ac.jp}\\
{Department of Mathematics, Osaka University}\\ 
{Toyonaka, Osaka 560-0043, Japan}
}
\date{}

\begin{document} 
\maketitle

\begin{abstract}
This paper explores a predictive game in which a Forecaster announces odds based on a time-homogeneous Markov kernel, establishing a game-theoretic law of large numbers for the relative frequencies of occurrences of all finite strings. A key feature of our proof is a betting strategy built on a universal coding scheme, inspired by the martingale convergence theorem and algorithmic randomness theory, without relying on a diversified betting approach that involves countably many operating accounts. We apply these insights to thermodynamics, offering a game-theoretic perspective on Le\'o Szil\'ard's thought experiment.


\end{abstract}

%==================================================================
\section{Introduction}\label{sec:Introduction}
%==================================================================

Game-theoretic probability theory \cite{itsonlyagame}, proposed by Shafer and Vovk in 2001, offers a framework for studying stochastic behavior without relying on traditional concept of probability.
To illustrate this approach, we begin by recalling a fundamental result from game-theoretic probability theory.

Let $\Omega := \{1, 2, \dots,A\}$ be a finite alphabet, and let $\Omega^n$, $\Omega^\ast$, and $\Omega^\infty$ denote the sets of sequences over $\Omega$ of length $n$, finite length, and  infinite (one-sided) length, respectively.
The empty string is denoted by $\lambda$.
An element of $\Omega^n$ is represented symbolically as $\omega^n$.
We also introduce the notation $\omega_i^j := \omega_i \omega_{i+1} \cdots \omega_j$ for $i\le j$, denoting the substring from the $i$th to the $j$th coordinates of a longer sequence $\omega_1\omega_2 \cdots \omega_n\cdots$.
By convention, if $i>j$, we set $\omega_i^j:= \lambda$.
For $x\in\Omega^*$ and $y\in\Omega^*\cup\Omega^\infty$, we write $x \sqsubset y$ to indicate that $x$ is a prefix of $y$.

Let us introduce the following set: 
\[
\mathcal{P}(\Omega):=\left\{ p:\Omega\to (0,1) \left| \; \sum_{\omega\in\Omega} p(\omega)=1 \right.\right\}.
\]
Given a $p\in\mathcal{P}(\Omega)$, consider the following game, where $\delta_{\omega_n}$ denotes the Kronecker delta, defined as $\delta_{\omega_n}(a)=1$ if $\omega_n=a$ and $\delta_{\omega_n}(a)=0$ otherwise.

\medskip
\begin{itembox}[l]{\bf Simple predictive game} \label{game:predictive_by_Shafer_and_Vovk}
\textbf{Players} : Skeptic and Reality. \\
\textbf{Protocol} : $K_0 = 1$. \\
\ \ FOR $n \in \mathbb{Z}_{>0}$ : \\
\ \ \ \ \ Skeptic announces $\beta_n \in \mathbb{R}^\Omega$. \\
\ \ \ \ \ Reality announces $\omega_n \in \Omega$. \\
\ \ \ \ \ $\displaystyle K_n := K_{n-1} +  \sum_{a \in \Omega}^{ } \beta_n (a) (\delta_{\omega_n} (a) - p(a))$. \\
\ \ END FOR.
\end{itembox}

\medskip
This protocol can be understood as a betting game in which Skeptic predicts Reality's ``stochastic'' move, regarding $p(a)$ as the ``probability'' of occurring $a\in\Omega$.
Further, $\beta_n$ and $K_n$ denote Skeptic's bet and capital at step $n$, respectively, with the recursion formula designating how the capital evolves.
Specifically, suppose that, at step $n$, Skeptic announces $\beta_n=(\beta_n(1), \beta_n(2),\dots, \beta_n(A))$ and Reality  announces $\omega_n=b\in\Omega$. 
Then, Skeptic obtains $\beta_n(b)(1-p(b))$ and loses $\sum_{a\neq b} \beta_n(a)p(a)$ in assets.
Note that $\beta_n$ can take negative values.
Since $\beta_n$ can depend on Reality's past move $\omega_1^{n-1}$, we identify Skeptic's strategy $\{\beta_n\}_n$ with a map $\beta: \Omega^\ast \rightarrow \mathbb{R}^\Omega$ as $\beta_n(a) := (\beta(\omega_1^{n-1})) (a).$

Apparently, this game is in favor of Reality because Reality announces $\omega_n$ {\it after} knowing Skeptic's bet $\beta_n$, preventing Skeptic from becoming rich.
However, Shafer and Vovk showed the following surprising result.

\begin{thm} [Game-theoretic law of large numbers] \label{SLLN:predictive-Shafer_and_Vovk}
In the simple predictive game, Skeptic has a prudent strategy $\beta : \Omega^\ast \rightarrow \mathbb{R}^\Omega$ that ensures $\lim_{n \to \infty} K_n = \infty$ unless
\[ 
\lim_{n \to \infty} \frac{1}{n}\sum_{i=1}^n \delta_a (\omega_i) = p(a)
\]
for all $a \in \Omega$.
Here, a strategy is called prudent if $K_n > 0$ for all $n\in\mathbb{Z}_{>0}$ and every sequence $\omega_1^n \in \Omega^n$ chosen by Reality.
\end{thm}

The theorem implies that there exists a betting strategy $\beta_n$ that guarantees Skeptic becomes infinitely rich if Reality's moves deviate from the ``law of large numbers,'' all while avoiding the risk of bankruptcy.

After Shafer and Vovk's original proof, which relies on a diversified betting approach using countably many operating accounts \cite{itsonlyagame}, several alternative proofs have been proposed \cite{Kumon-Takemura-Takeuchi, Takeuchi-Kumon-Takemura}.
The objective of this paper is to elucidate the coding theoretic aspect underlying this protocol by providing an alternative proof of Theorem \ref{SLLN:predictive-Shafer_and_Vovk} and its generalizations.

In order to explicate our approach, consider the following generalized game in which a Forecaster comes into play to announce a possibly ``non-i.i.d.'' process.

\medskip
\begin{itembox}[l]{\bf Generalized predictive game} \label{game:predictive} 
\textbf{Players} : Forecaster, Skeptic, and Reality. \\
\textbf{Protocol} : $K_0 = 1$. \\
\ \ FOR $n \in \mathbb{Z}_{>0}$ : \\
\ \ \ \ \ Forecaster announces $p_n \in \mathcal{P}(\Omega)$. \\
 \ \ \ \ \ Skeptic announces $\beta_n \in \mathbb{R}^\Omega$. \\
\ \ \ \ \ Reality announces $\omega_n \in \Omega$. \\
\ \ \ \ \ $\displaystyle K_n := K_{n-1} +  \sum_{a \in \Omega}^{ } \beta_n (a) (\delta_{\omega_n} (a) - p_n(a))$. \\
\ \ END FOR.
\end{itembox}

\medskip
Let us identify Skeptic's betting strategy $\beta_n\in\mathbb{R}^\Omega$ with $\alpha_n\in\mathbb{R}^\Omega$ that satisfies
\[ \beta_n(a)=K_{n-1} \cdot \alpha_n(a),\qquad (a\in\Omega). \]
Then, the recursion formula for the capital is rewritten as
\begin{equation}\label{eqn:recursionKn}
 K_n = K_{n-1}\left\{ 1 +  \sum_{a \in \Omega} \alpha_n (a) (\delta_{\omega_n} (a) - p_n(a)) \right\}.
\end{equation}
We shall call $\alpha_n$ a betting strategy as well, and call it prudent if the corresponding $\beta_n$ is prudent.
We also identify $\{\alpha_n\}_n$ with a map $\alpha:\Omega^\ast \rightarrow \mathbb{R}^\Omega$ as $\alpha_n(a) = (\alpha(\omega_1^{n-1}))(a)$.

Now, associated with a prudent strategy $\alpha_n$ is the following quantity: 
\begin{equation}\label{eqn:probQ}
 Q(\omega \mid \omega_1^{n-1})
 := \left\{ 1 + \sum_{a \in \Omega} \alpha_n(a) (\delta_\omega (a) - p_n(a)) \right\}p_n(\omega).
\end{equation}
Since $\alpha_n$ is prudent, we see that $Q(\omega \mid \omega_1^{n-1}) > 0$ for all $\omega\in\Omega$. Moreover,
\begin{align*}
\sum_{\omega \in \Omega} Q(\omega \mid \omega_1^{n-1})
&= \sum_{\omega \in \Omega} \left\{ 1 + \sum_{a \in \Omega} \alpha_n(a) (\delta_\omega (a) - p_n(a)) \right\}p_n(\omega) \\
&= 1 + \sum_{a \in \Omega} \alpha_n(a) \left\{ \sum_{\omega \in \Omega} (\delta_\omega(a) - p_n(a)) p_n(\omega) \right\} \\
&= 1 + \sum_{a \in \Omega} \alpha_n(a) (p_n(a) - p_n(a)) = 1.
\end{align*}
Therefore, the quantity $Q(\omega \mid \omega_1^{n-1})$ defined by \eqref{eqn:probQ} can be regarded as a conditional probability. 
Conversely, for any conditional probability $Q(\omega \mid \omega_1^{n-1})$, 
there exists a prudent strategy $\alpha_n$ (although not unique) that satisfies \eqref{eqn:probQ}: 
for instance, let $\alpha_n(a) := Q(a \mid \omega_1^{n-1}) / p_n(a)$ for each $a \in \Omega$.
Thus, the role of Skeptic in the above predictive game is regarded as announcing a conditional probability $Q(\omega \mid \omega_1^{n-1})$.

Now, suppose that Forecaster happens to have a predetermined probability measure $P$ on $(\Omega^\infty, \mathscr{F})$, with $\mathscr{F}:=\sigma(\{\Gamma_x\}_{x\in\Omega^*})$ being the $\sigma$-algebra generated by the cylinder sets $\Gamma_x := \{y \in \Omega^\infty : x \sqsubset y\}$,
and announces each function $p_n$ as the conditional probability, given the past data $\omega_1^{n-1}$, as follows:
\[ 
p_n(a) := P(a \mid \omega_1^{n-1}) :=P(a \mid \Gamma_{\omega_1^{n-1}}), \quad (a\in\Omega).
\]
Then, we have from \eqref{eqn:recursionKn} and \eqref{eqn:probQ} that
\begin{equation*}
\frac{K_n}{K_{n-1}}
= 1 +  \sum_{a \in \Omega} \alpha_n (a) (\delta_{\omega_n} (a) - p_n(a)) 
= \frac{Q(\omega_n \mid \omega_1^{n-1})}{P(\omega_n \mid \omega_1^{n-1})},
\end{equation*}
and hence
\begin{equation}\label{eqn:capital}
K_n = K_0 \, \prod_{i=1}^n \frac{K_i}{K_{i-1}}=\frac{Q(\omega_1^n)}{P(\omega_1^n)}.
\end{equation}
Put differently, the capital process $K_n$ is nothing but the likelihood ratio process between $P$ and $Q$. 

Let $\mathscr{F}_n := \sigma(\{\Gamma_{x^n}\}_{x^n \in \Omega^n})$ for each $n\in \mathbb{Z}_{>0}$,
Then the capital process \eqref{eqn:capital} is a $P$-martingale relative to the natural filtration $\{\mathscr{F}_n\}_n$, in that
\begin{eqnarray*}
\mathbb{E} \left[\left. \frac{Q(\omega_1^n)}{P(\omega_1^n)}\, \right| \mathscr{F}_{n-1} \right]
= \sum_{\omega_n \in \Omega} \frac{Q(\omega_1^n)}{P(\omega_1^n)} P(\omega_n \mid \omega_1^{n-1}) 
= \sum_{\omega_n \in \Omega} \frac{Q(\omega_1^n)}{P(\omega_1^{n-1})}
= \frac{Q(\omega_1^{n-1})}{P(\omega_1^{n-1})}.
\end{eqnarray*}
It then follows from the martingale converge theorem \cite{prob-with-martingales} that the capital process $K_n$ converges almost surely to some nonnegative value under the probability measure $P$. 
Game-theoretic probability provides a reciprocal description of this mechanism by asserting that $K_n$ diverges ``on a null set with respect to $P$,'' which, in the context of predictive games, is interpreted as ``if Reality does not align with Forecaster $P$.''

Furthermore, due to \eqref{eqn:capital}, the logarithm of the capital process is written as 
\begin{equation}\label{eqn:logKn}
\log_A K_n=\left\{-\log_A P(\omega_1^n)\right\}-\left\{-\log_A Q(\omega_1^n)\right\}.
\end{equation}
This is simply the difference between the Shannon codelengths for $P$ and $Q$. 
In other words, designing a betting strategy $\alpha_n$ via the conditional probability \eqref{eqn:probQ} is equivalent to designing a coding scheme for a given data $\omega_1^{n-1}$ via $Q$.

Note that the quantity \eqref{eqn:logKn} reminds us of the randomness deficiency in algorithmic randomness theory \cite{LiVitanyi}. 
For a computable probability measure $P$ on $\Omega^\infty$, an infinite sequence $\omega_1^\infty \in \Omega^\infty$ is Martin-L\"{o}f $P$-random if and only if the sequence
\begin{equation}\label{eqn:KC}
-\log_A P(\omega_1^n) - \mathcal{K}(\omega_1^n)
\end{equation}
of randomness deficiencies is bounded from above, where $\mathcal{K}(\omega_1^n)$ is the prefix Kolmogorov complexity of $\omega_1^n \in \Omega^\ast$.
Obviously, \eqref{eqn:logKn} and \eqref{eqn:KC} are similar in forms, as they are differences of  codelengths.
It is, however, crucial to observe that \eqref{eqn:KC} contains an uncomputable quantity $ \mathcal{K}(\omega_1^n)$. 
This observation prompts us to design the conditional probability $Q$ in \eqref{eqn:probQ}, or equivalently a betting strategy $\alpha_n$ in the generalized predictive game, in terms of a computable universal coding scheme.
This motivates our approach in the present study.

In this paper, we treat the following predictive game, which we shall call a {\it time-homogeneous Markovian predictive game}%
\footnote{
The Markovian predictive game introduced in this paper is entirely distinct from the Markov game commonly used in the field of operations research \cite{Solan}.
}.
Suppose Forecaster has a $k$th-order Markov kernel $M : \Omega\times \Omega^k \to (0,1): (a, \omega^k) \mapsto M(a \mid \omega^k)$ that satisfies
\[ \sum_{a \in \Omega} M(a \mid \omega^k) = 1 \]
for all $\omega^k \in \Omega^k$.

\begin{itembox}[l]{\bf Time-homogeneous $k$th-order Markovian predictive game} \label{game:Markov}
\textbf{Players} : Forecaster, Skeptic, and Reality. \\
\textbf{Protocol} : $K_0 = 1$.\\
\ \ FOR $n \in \mathbb{Z}_{>0}$ : \\
\ \ \ \ \ Forecaster announces $p_n\in\mathcal{P}(\Omega)$ such that 
	 $\displaystyle p_n(a):=M(a \mid \omega_{n-k}^{n-1})$ for $n>k$.  \\
\ \ \ \ \ Skeptic announces $\alpha_n \in \mathbb{R}^\Omega$. \\
\ \ \ \ \ Reality announces $\omega_n \in \Omega$. \\
\ \ \ \ \ $\displaystyle K_n := K_{n-1} \left\{ 1 +  \sum_{a \in \Omega} \alpha_n (a) (\delta_{\omega_n} (a) - p_n(a)) \right\}$. \\
\ \ END FOR.
\end{itembox}

This protocol can be understood as a variant of the generalized predictive game in which Forecaster announces $p_n :\Omega\to (0,1)$ for $n>k$ according to a time-homogeneous Markov kernel $M$ as $p_n(a)=M(a \mid \omega_{n-k}^{n-1})$ based on Reality's past moves.

For each $\omega^\ell\in\Omega^\ell$ with $\ell \in \mathbb{Z}_{\geq k}$, let $P(\omega^\ell)$ be defined by
\[
P(\omega^\ell):=\pi(\omega^k) \prod_{i=k+1}^\ell M(\omega_i \mid \omega^{i-1}_{i-k}),
\]
where $\pi:\Omega^k\to (0,1)$ is the stationary distribution associated with the Markov kernel $M$.
The main result of this paper is the following:

\begin{thm}\label{thm:kth_Markov_normal_number}
In the time-homogeneous $k$th-order Markovian predictive game,
Skeptic has a prudent strategy $\alpha : \Omega^\ast \rightarrow \mathbb{R}^\Omega$ that ensures $\lim_{n \to \infty} K_n = \infty$ unless
\[ 
\lim_{n \to \infty} \frac{S_n(a^\ell)}{n} = P(a^\ell)
\]
for all $\ell\in \mathbb{Z}_{\geq k}$ and $a^\ell \in\Omega^\ell$,
where $S_n(a^\ell)$ denotes the number of occurrences of $a^\ell$ in $\omega^n$.
\end{thm}

Theorem \ref{thm:kth_Markov_normal_number} establishes that there exists a betting strategy that guarantees Skeptic can become infinitely rich if Reality's moves do not align with the Markovian Forecaster's announcements. Specifically, this happens when the relative frequency of occurrences of some string of length $\ell\;(\ge k)$ fails to converge to the stationary joint distribution associated with the $k$th-order Markov kernel.

This paper is organized as follows.
In Section 2, we introduce a betting strategy with an emphasis on its universal coding theoretic aspect and present several lemmas to lay the groundwork for proving the main result.
For improved readability, all proofs of lemmas are deferred to Appendix \ref{app:Lemmas}.
In Section 3, we present a proof of Theorem \ref{thm:kth_Markov_normal_number}.
Section 4 explores applications of Theorem \ref{thm:kth_Markov_normal_number} to thermodynamics, specifically a game-theoretic interpretation of Szil\'ard's engine and a discussion of entropy in predictive games. 
Section 5 provides concluding remarks.
For the reader's convenience, additional information on stationary distributions of Markov chains and an alternative proof of Theorem \ref{SLLN:predictive-Shafer_and_Vovk} are provided in Appendices \ref{app:stationary} and \ref{app:Lynch-Davisson}, respectively.

%==================================================================
\section{Preliminaries}
%==================================================================

In this section, we develop a betting strategy using the technique of incremental parsing and establish several lemmas in preparation for the proof of Theorem \ref{thm:kth_Markov_normal_number}. 

%-------------------------------------------------------------------------------------------------------------------
\subsection{Betting strategy inspired by Lempel-Ziv coding scheme}
%-------------------------------------------------------------------------------------------------------------------

We outline an algorithm for incremental parsing \cite{LZ-78}, which divides a string into substrings separated by slashes, with each substring being the shortest one not previously encountered.
The algorithm runs as follows:
Start with an initial slash.
After each slash, scan the input sequence until the shortest string that has not yet been marked off is identified.
Since this string is the shortest unseen string, all its prefixes must have appeared earlier in the sequence.
For example, a sequence $1000011101011$ of length 13 is decomposed into 
\[ /1/0/00/01/11/010/11. \]
Suppose a sequence $\omega_1^n$ is parsed as
\[
\omega_1^n = /\omega_{n_0 + 1}^{n_1}/ \omega_{n_1 + 1}^{n_2}/ \cdots /\omega_{n_{T-1} +1}^{n_T} /\omega_{n_T +1}^n,
\]
where $n_0 = 0$, $n_1 = 1$, and all parsed substrings expect the last one, $\omega_{n_T +1}^n$, are distinct.
In what follows, the substrings $\omega_{n_0 + 1}^{n_1}, \omega_{n_1 + 1}^{n_2}, \dots, \omega_{n_{T-1} +1}^{n_T}$ are referred to as parsed phrases. 
Note that the number $T$ of parsed phrases depends on the sequence $\omega_1^n$, 
and the last string $\omega_{n_T+1}^n$ may be empty.

We now construct a betting strategy at step $n$, given Reality's past moves $\omega_1^{n-1}$ ($n \geq 2$).
Using incremental parsing, we decompose $\omega_1^{n-1}$ into 
\[ 
\omega_1^{n-1} =/ \omega_{n_0 + 1}^{n_1} /\omega_{n_1 + 1}^{n_2} /\cdots /\omega_{n_{T-1} +1}^{n_T}/ \omega_{n_T +1}^{n-1}.
\]
Next, we define the set
\begin{equation*}
V(\omega_1^{n_T}) := 
\left\{ \xi \in \Omega^\ast \middle|
  \begin{alignedat}{2}
  \; &\xi \neq \omega_{n_{j-1} + 1}^{n_j} \,\text{ for all } j\in \{0, 1, \dots, T\}, \,\text{and} &\\
  & \xi =   \omega_{n_{j-1} + 1}^{n_j} b \,\text{ for some } j\in \{0, 1, \dots, T\} \,\text{and}\, b \in \Omega &
  \end{alignedat}
\right\},
\end{equation*}
where $\omega_{n_{-1}+1}^{n_0} = \lambda$ is the empty string.
The set $V(\omega_1^{n_T})$ is a prefix set containing all potential phrases that can be the next parsed phrase following $\omega_{n_{T-1} +1}^{n_T}$.
The size of this set is given by 
\[ |V(\omega_1^{n_T})| = A + T(A-1). \] 
This can be shown by induction on $T$: 
For $T = 0$, we have $n= 1$, and
\begin{equation*}
V(\omega_1^{n-1}) = V(\lambda) = \Omega.
\end{equation*}
For $T \geq 1$, the set $V(\omega_1^{n_T})$ is constructed as
\begin{eqnarray*}
V(\omega_1^{n_T})
= \left( V(\omega_1^{n_{T-1}}) \setminus \{\omega_{n_{T-1} +1}^{n_T}\} \right) \cup \{\omega_{n_{T-1}+1}^{n_T} b \mid b \in \Omega\},
\end{eqnarray*}
which yields a recursive formula $|V(\omega_1^{n_T})| = |V(\omega_1^{n_{T-1}})| -1 + A$, ensuring the desired result.

Finally, we define the conditional probability $Q_{LZ}(\omega_n \mid \omega_1^{n-1})$, which determines the betting strategy $\alpha_n(\omega_n)$, as follows.
For $n= 1$, let $Q_{LZ}(a) : = 1/A$. For $n \geq 2$, we define
\begin{equation}\label{eqn:Q_LZ}
Q_{LZ}(a \mid \omega_1^{n-1}) :=
\displaystyle \frac{|\{ \xi \in V(\omega_1^{n_T}) \mid \omega_{n_T +1}^{n-1} a \sqsubset \xi \}|}{|\{ \xi \in V(\omega_1^{n_T}) \mid \omega_{n_T +1}^{n-1} \sqsubset \xi\}|}.
\end{equation}
Note that
\[
\{ \xi \in V(\omega_1^{n_T}) \mid \omega_{n_T +1}^{n-1} \sqsubset \xi \} = \bigsqcup_{a \in \Omega} \{ \xi \in V(\omega_1^{n_T}) \mid \omega_{n_T +1}^{n-1} a \sqsubset \xi \},
\]
which follows from the definition of $V(\omega_1^{n_T})$. Moreover,
the set 
\[ \{ \xi \in V(\omega_1^{n_T}) \mid \omega_{n_T +1}^{n-1} a \sqsubset \xi \} \]
is nonempty for any $a\in\Omega$.
Thus, we conclude that 
\[ \sum_{a \in \Omega} Q_{LZ}(a \mid \omega_1^{n-1}) = 1 \quad \mbox{and}\quad Q_{LZ}(a \mid \omega_1^{n-1})>0. \]

The motivation behind the definition \eqref{eqn:Q_LZ} is now in order. 
When $n-1 = n_T$ (i.e., when $\omega_{n_T +1}^{n-1} = \lambda$), we have
\[
Q_{LZ}(a \mid \omega_1^{n_T}) = \displaystyle \frac{|\{ \xi \in V(\omega_1^{n_T}) \mid a \sqsubset \xi \}|}{|V(\omega_1^{n_T})|}.
\]
Thus, for each $\xi=\xi_1^\tau \in V(\omega_1^{n_T})$,
\begin{align*}
Q_{LZ}(\xi \mid \omega_1^{n_T}) 
&:= Q_{LZ}(\xi_1 \mid \omega_1^{n_T}) \cdot Q_{LZ}(\xi_2 \mid \omega_1^{n_T} \xi_1)
\cdots Q_{LZ}(\xi_\tau \mid \omega_1^{n_T} \xi_1^{\tau-1})
= \frac{1}{|V(\omega_1^{n_T})|}.
\end{align*}
In other words, the conditional probability $Q_{LZ}(a \mid \omega_1^{n-1})$ is designed to induce the uniform distribution over $V(\omega_1^{n_T})$.
This observation also implies that  
\begin{equation}\label{eqn:evalQ_LZ}
 Q_{LZ}(\omega_{n_T +1}^n \mid \omega_1^{n_T}) > \frac{1}{|V(\omega_1^{n_T})|} 
\end{equation}
whenever $n>n_T$.

In what follows, we shall call the betting strategy based on the conditional probability $Q_{LZ}(a \mid \omega_1^{n-1})$ the {\it Lempel-Ziv betting strategy}.

\begin{rem}
From the above discussion,
\begin{equation}\label{eqn:LZcodelength}
 -\log Q_{LZ}(\omega_1^{n_T}) = \sum_{j= 0}^{T-1} \log (A + j(A-1)),
\end{equation}
where
\[ Q_{LZ}(\omega_1^n):=\prod_{i=1}^n Q_{LZ}(\omega_i \mid \omega_1^{i-1}). \] 
On the other hand, the Lempel-Ziv codelength $\ell_{LZ}$ \cite{LZ-78} is given by
\[ \ell_{LZ}(\omega_1^n) = \sum_{j = 1}^{T+1} \lceil \log_A (j A) \rceil. \]
\end{rem}

\subsection{Properties of Lempel-Ziv betting strategy}\label{subsec:lemma}

In this section, we outline several fundamental properties of the Lempel-Ziv betting strategy. All  proofs are deferred to Appendix \ref{app:Lemmas}.

Define the complexity $c(\omega_1^n)$ of a sequence $\omega_1^n \in \Omega^n$ as the total number $T$ of parsed phrases obtained through the incremental parsing of $\omega_1^n$ \cite[p.~448]{CoverThomas}.

\begin{lem}\label{lem:liminf_comp_and_Q}
For any $\omega_1^\infty \in \Omega^\infty$, the following inequality holds:
\[
\liminf_{n \to \infty} \frac{1}{n} \left\{ 
c(\omega_1^n)\log c(\omega_1^n) - (-\log Q_{LZ}(\omega_1^n)) \right\} \geq 0. 
\]
\end{lem}

For $n, \ell \in \mathbb{Z}_{>0}$ with $n>\ell$, let $T_n(a_1^\ell)$ denote the number of occurrences of $a_1^\ell \in \Omega^\ell$ in the cyclically extended word $\omega_1^n \omega_1^{\ell-1}$ of length $n+\ell-1$.
Similarly, let $T_n(b \mid a_1^\ell)$ represent the number of occurrences of $b$ immediately following $a_1^\ell \in \Omega^\ell$ in the extended word $\omega_1^n \omega_1^\ell$ of length $n+\ell$.
In other words, $T_n(b\mid a_1^\ell)=T_n(a_1^\ell b)$.

Note that, by definition,
\[ \sum_{a_1^\ell \in \Omega^\ell} T_n(a_1^\ell) = n. \]
Moreover, the quantity $T_n(a_1^\ell)$ is asymptotically equivalent to $S_n(a_1^\ell)$, the number of occurrences of $a_1^\ell$ in $\omega_1^n$, in the sense that
\[ \lim_{n \to \infty} \frac{1}{n}\left( T_n(a_1^\ell) - S_n(a_1^\ell) \right) = 0. \]
The following lemma will prove useful in the sequel.

\begin{lem}\label{lem:cond_empirical}
Let $n, \ell \in \mathbb{Z}_{>0}$ with $n>\ell$. Then, for all $a_1^\ell \in \Omega^\ell$ and $b \in \Omega$, the following identities hold:
\[
\sum_{b \in \Omega} T_n(b \mid a_1^\ell) = T_n(a_1^\ell) \quad\mbox{and}\quad
\sum_{a_1 \in \Omega} T_n(b \mid a_1^\ell) = T_n(a_2^\ell b).
\]
\end{lem}

Let us now fix $\ell \in \mathbb{Z}_{>0}$ arbitrarily.
Given $\omega_1^n$ with $n >\ell$, we define
\[
 \hat{M}_n^\ell (b \mid a^\ell) := \frac{T_n(b \mid a^\ell)}{T_n(a^\ell)}
\]
for $a^\ell \in \Omega^\ell$ satisfying $T_n(a^\ell)>0$. 
Due to Lemma \ref{lem:cond_empirical}, we observe that 
\[
\sum_{b \in \Omega} \hat{M}_n^\ell(b \mid a^\ell) = \sum_{b \in \Omega} \frac{T_n(b \mid a^\ell)}{T_n(a^\ell)} = 1.
\]
Thus, with an appropriate definition of $\hat{M}_n^\ell (b \mid a^\ell)$ for $a^\ell\in \Omega^\ell$ satisfying $T_n(a^\ell)=0$, $\hat{M}_n^\ell$ can be regarded as an $\ell$th-order Markov kernel formally associated with $\omega_1^n$. 
The following lemma provides a condition that ensures $T_n(a^k)/n$ converges to the stationary distribution $\pi$ of the $k$th-order Markov kernel $M$.

\begin{lem} \label{lem:kth_Markov_stationary}
If
\[
\lim_{n\to\infty} \hat{M}_n^k(b \mid a_1^k)
= M(b \mid a_1^k)
\]
for all $a_1^k \in \Omega^k$ and $b \in \Omega$, then,
\[
\lim_{n \to \infty} \frac{T_n(a_1^k)}{n} = \pi(a_1^k)
\]
for all $a_1^k \in \Omega^k$. 
\end{lem}

Next, for $n, \ell \in \mathbb{Z}_{>0}$ with $n>\ell$, we define
\[
  \hat{R}_n^{\ell}(\omega_1^n) :=\left\{\prod_{i=1}^{\ell} \hat{M}_n^\ell(\omega_i \mid \omega_{n-\ell+i}^n \omega_1^{i-1}) \right\}
    \cdot \left\{\prod_{i= \ell+1}^n \hat{M}_n^\ell(\omega_i\mid \omega_{i-\ell}^{i-1}) \right\}.
\]
In the first factor, which represents $\hat{R}_n^{\ell}(\omega_1^\ell)$, we formally introduce concatenated strings $\omega_{n-\ell+i}^n \omega_1^{i-1}$ of length $\ell$ for $i=1,\dots, \ell$, to facilitate the application of the $\ell$th Markov kernel $\hat{M}_n^\ell$.

\begin{lem}
\label{lem:Ziv_ineq}
For any $n, \ell \in \mathbb{Z}_{>0}$ with $n>\ell$ and $\omega_1^n \in \Omega^n$, the following inequality holds:
\[
\frac{1}{n} c(\omega_1^n) \log c(\omega_1^n) \leq -\frac{1}{n}\log \hat{R}_n^\ell(\omega_1^n) + \delta_\ell(n),
\]
where $\delta_\ell(n) \rightarrow 0$ as $n \rightarrow \infty$.
\end{lem}

Finally, for $n \in \mathbb{Z}_{\geq k}$, let us introduce
\[
\tilde{P}(\omega_1^n) :=\left\{\prod_{i= 1}^k p_i(\omega_i)\right\}\cdot \left\{ \prod_{i=k+1}^n M(\omega_i \mid \omega^{i-1}_{i-k}) \right\}, 
\]
which represents Forecaster's announcements.
The next lemma establishes the relationship between $\tilde{P}$ and $\hat{R}_n^\ell$.

\begin{lem}\label{lem:lth_Markov_divergence}
For any $n, \ell \in \mathbb{Z}_{\ge k}$ with $n>\ell$, the following identity holds:
\begin{align*}
-\log \tilde{P}(\omega_1^n) + \log \hat{R}^{\ell}_n (\omega_1^n)
&= -\log \tilde{P}(\omega_1^{\ell}) + \log \prod_{i= 1}^{\ell} M(\omega_i \mid \omega_{n-k+i}^n \omega_1^{i-1})  \\
&\quad
+ \sum_{a_1^{\ell} \in \Omega^{\ell}} T_n(a_1^{\ell}) \cdot 
 D \left(\left. \hat{M}_n^\ell(\;\cdot \mid a_1^\ell) \right\| M(\;\cdot \mid a_{\ell-k+1}^{\ell}) \right),
\end{align*}
where $D(\,\cdot\,  || \,\cdot\,)$ is the Kullback-Leibler divergence.
\end{lem}

%-------------------------------------------------------------------------------------------------------------------
\section{Proof of Theorem \ref{thm:kth_Markov_normal_number}}
%-------------------------------------------------------------------------------------------------------------------

Let $K^\alpha: \Omega^* \to\mathbb{R}$ represent the capital process associated with a betting strategy $\alpha$, so that $K_n=K^\alpha(\omega_1^n)$ for $\omega_1^n\in\Omega^*$.
The following lemma demonstrates that the limit supremum for the capital process can be replaced by a limit.

\begin{lem} \label{lem:limsup}
Suppose a strategy $\alpha$ satisfies 
\[ \limsup_{n \to \infty} K^\alpha(\omega_1^n) = \infty \]
for a specific sequence of Reality's moves $\{\omega_1^n\}_{n \in \mathbb{Z}_{>0}}$. 
Then there exists another strategy $\alpha_*$ such that 
\[ \lim_{n \to \infty} K^{\alpha_*} (\omega_1^n) = \infty \]
for the same sequence $\{\omega_1^n\}_{n \in \mathbb{Z}_{>0}}$.
\end{lem}

\begin{proof}
The required strategy $\alpha_*$ can be defined as follows:
Strategy $\alpha_*$ uses $\alpha$ as long as the capital $K_n$ retains below 2. 
Once $K_n$ reaches or exceeds $2$, $\alpha_*$ transfers the net gain $\Delta K:=K_n-1\, (\ge 1)$ into an external storage. It then restart the game with a capital of $1$, employing the strategy $\alpha$ again. 
\end{proof}

\begin{lem} \label{lem:SLLN:kth_Markov}
In the time-homogeneous $k$th-order Markovian predictive game, the Lempel-Ziv betting strategy $Q_{LZ}$ ensures $\lim_{n \to \infty} K_n = \infty$ unless 
\[ 
\lim_{n \to \infty} \hat{M}_n^k(b \mid a_1^k) = M(b \mid a_1^k)
\]
for all $a_1^k \in \Omega^k$ and  $b \in \Omega$.
\end{lem}

\begin{proof}
Using equation \eqref{eqn:capital} and applying Lemma \ref{lem:Ziv_ineq}, we obtain
\begin{align*}
\frac{\log K_n}{n}
&= \frac{1}{n} \left\{-\log \tilde{P}(\omega_1^n)  - (-\log Q_{LZ}(\omega_1^n)) \right\} \\
&= \frac{1}{n} \left\{-\log \tilde{P}(\omega_1^n) - c(\omega_1^n)\log c(\omega_1^n) \right\}+ \frac{1}{n} \Big\{ c(\omega_1^n)\log c(\omega_1^n) -( -\log Q_{LZ}(\omega_1^n)) \Big\} \\
&\geq \frac{1}{n} \left\{-\log \tilde{P}(\omega_1^n) + \log \hat{R}_n^{k}(\omega_1^n)- n\delta_k(n) \right\} 
+ \frac{1}{n} \Big\{ c(\omega_1^n)\log c(\omega_1^n) - (-\log Q_{LZ}(\omega_1^n)) \Big\}.
\end{align*}
Applying Lemma \ref{lem:lth_Markov_divergence} with $\ell=k$, we further evaluate $K_n$ as
\begin{align*}
\frac{\log K_n}{n}
&\ge \sum_{a_1^k \in \Omega^k} \frac{T_n(a_1^k)}{n} \cdot 
 D\left( \left. \hat{M}_n^k(\;\cdot \mid a_1^k) \right\|  M(\;\cdot \mid a_1^k) \right) \\
&\ \ \ + \frac{1}{n} \left\{ -\log \tilde{P}(\omega_1^k) + \log \prod_{i= 1}^{k} M(\omega_i \mid \omega_{n-k+i}^n \omega_1^{i-1}) \right\} -  \delta_{k} (n) \\
&\ \ \ +\frac{1}{n} \Big\{ c(\omega_1^n) \log c(\omega_1^n) - (-\log Q_{LZ}(\omega_1^n)) \Big\}.
\end{align*}
From the definition of $\tilde{P}$,
\[
\lim_{n \to \infty}\frac{1}{n} \left\{ -\log \tilde{P}(\omega_1^{k}) + \log \prod_{i= 1}^{k} M(\omega_i \mid \omega_{n-k+i}^n \omega_1^{i-1}) \right\} =0.
\]
Furthermore, by Lemma \ref{lem:Ziv_ineq}, we know that $\delta_{k}(n) \to 0$ as $n \to \infty$, 
and by Lemma \ref{lem:liminf_comp_and_Q},
\[
\liminf_{n \to \infty}\frac{1}{n} \Big\{ c(\omega_1^n) \log c(\omega_1^n) - (-\log Q_{LZ}(\omega_1^n)) \Big\} \geq 0.
\]

In light of Lemma \ref{lem:limsup}, therefore, it now suffices to prove the following claim:
For a given $\omega_1^\infty \in \Omega^\infty$, if there exists $b_1^{k+1} \in \Omega^{k+1}$ such that
$\hat{M}_n^k(b_{k+1} \mid b_1^k)$ does not converge to $M(b_{k+1} \mid b_1^k)$ as $n\to\infty$,
then
\[
\limsup_{n \to \infty} \sum_{a_1^k \in \Omega^k} \frac{T_n(a_1^k)}{n} \cdot D\left( \left. \hat{M}_n^k(\;\cdot \mid a_1^k) \right\| M(\;\cdot \mid a_1^k) \right) > 0.
\]

We prove this claim by contradiction.
Assume that for a given $\omega_1^\infty \in \Omega^\infty$, there exists $b_1^{k+1} \in \Omega^{k+1}$ such that
$\hat{M}_n^k(b_{k+1} \mid b_1^k)$ does not converge to $M(b_{k+1} \mid b_1^k)$ as $n\to\infty$, and yet 
\begin{equation}\label{eqn:div_converge_0}
\lim_{n \to \infty} 
\sum_{a_1^k \in \Omega^k} 
\frac{T_n(a_1^k)}{n} \cdot D\left(\left. \hat{M}_n^k(\;\cdot \mid a_1^k) \right\| M(\;\cdot \mid a_1^k) \right) = 0. 
\end{equation}
From this assumption, there exists a subsequence $(n_i)_i\subset (n)$ such that
\begin{equation}\label{eqn:induction1}
  \lim_{i \to\infty} \frac{T_{n_i}(b_1^k)}{n_i} = 0.
\end{equation}
Further, using \eqref{eqn:div_converge_0} and \eqref{eqn:induction1}, we can deduce that 
for all $d_k\in\Omega$, 
\begin{equation}\label{eqn:induction2}
  \lim_{i \to\infty} \frac{T_{n_i}(d_k b_1^{k-1})}{n_i} = 0.
\end{equation}
To establish \eqref{eqn:induction2}, we first recall the second identity in Lemma \ref{lem:cond_empirical}, which yields
\[
 T_n(d a_2^\ell b)\le \sum_{a_1\in\Omega}T_n(a_1 a_2^\ell b) = T_n(a_2^\ell b)
\]
for any $d\in\Omega$ and $a_2^\ell b\in\Omega^{\ell}$.
Thus, for any $d_k\in\Omega$, 
\[
 0 \leq \frac{T_{n_i}(d_k b_1^k)}{n_i}  \leq \frac{T_{n_i}(b_1^k)}{n_i}.
\]
From \eqref{eqn:induction1}, it follows that
\begin{equation}\label{eqn:induction2.1}
 \lim_{i \to \infty} \frac{T_{n_i}(d_k b_1^k)}{n_i} = 0. 
\end{equation}
Let us prove \eqref{eqn:induction2} by contradiction. Suppose there exists $d_k \in \Omega$ such that
\[
\limsup_{i \to \infty} \frac{T_{n_i}(d_k b_1^{k-1})}{n_i} >0.
\]
Then, there exists a subsequence $(n_{i_j})_{j} \subset (n_i)_i$ such that
\begin{equation}\label{eqn:induction2.2}
 \lim_{j \to \infty} \frac{T_{n_{i_j}}(d_k b_1^{k-1}) }{ n_{i_j}}  >0.
\end{equation}
Consequently, for sufficiently large $j$, we have
\begin{align*}
\frac{T_{n_{i_j}}( d_k b_1^k)}{n_{i_j}}
&= \frac{T_{n_{i_j}}(b_k \mid d_k b_1^{k-1})}{n_{i_j}}
= \frac{T_{n_{i_j}}(d_k b_1^{k-1})}{n_{i_j}} \frac{T_{n_{i_j}}(b_k \mid d_k b_1^{k-1})}{T_{n_{i_j}}(d_k b_1^{k-1})} \\
&= \frac{T_{n_{i_j}}(d_k b_1^{k-1})}{n_{i_j}} \hat{M}_{n_{i_j}}^k(b_k \mid d_k b_1^{k-1}),
\end{align*}
and thus, combining \eqref{eqn:induction2.1} and \eqref{eqn:induction2.2}, we find $\lim_{j \to \infty} \hat{M}_{n_{i_j}}^k(b_k \mid d_k b_1^{k-1}) =0$. 
Recalling $M(b_k \mid d_k b_1^{k-1})>0$, we find that
\begin{equation}\label{eqn:induction2.3}
\lim_{j \to \infty} \hat{M}_{n_{i_j}}^k(b_k \mid d_k b_1^{k-1}) \neq M(b_k \mid d_k b_1^{k-1}).
\end{equation}
Since \eqref{eqn:induction2.2} and \eqref{eqn:induction2.3} contradict \eqref{eqn:div_converge_0}, 
we have \eqref{eqn:induction2}.

By repeatedly applying this reasoning, adding $d_{k-i+1}$ to the front and removing $b_{k-i+1}$ from the rear of the word $b_1^k$ in \eqref{eqn:induction1} for $i=1,\dots,k$, 
we conclude that
\[
  \lim_{i \to\infty} \frac{T_{n_i}(d_1^k)}{n_i} = 0
\]
for all $d_1^k\in\Omega^k$.  
This contradicts the fact that
\[ 
 \sum_{d_1^k\in\Omega} T_{n}(d_1^k)=n
\]
for all $n \;(\geq k)$. Hence, the claim is proved.
\end{proof}

Now we are ready to prove Theorem \ref{thm:kth_Markov_normal_number}.

\begin{proof}[Proof of Theorem \ref{thm:kth_Markov_normal_number}]
We prove the assertion by induction in $\ell\;(\geq k)$.
For $\ell = k$, the assertion holds by Lemmas \ref{lem:kth_Markov_stationary} and \ref{lem:SLLN:kth_Markov}.

Assume that the assertion holds for some $\ell\; (\ge k)$, namely,
\[
\lim_{n \to \infty} \frac{T_n(a_1^{\ell})}{n} = P(a_1^{\ell})
=\pi(a_1^k) \prod_{i=k+1}^{\ell} M(a_i \mid a^{i-1}_{i-k}) \;(>0) 
\]
for all $a_1^{\ell} \in \Omega^{\ell}$.
Since
\[
\frac{T_n(a_1^{\ell+1})}{n} 
= \frac{T_n(a_1^{\ell})}{n} \hat{M}_n^{\ell}(a_{\ell+1} \mid a_1^{\ell}), 
\]
we see that 
\[ \lim_{n\to\infty}\frac{T_n(a_1^{\ell+1})}{n} =P(a_1^{\ell+1})\]
holds if and only if 
\[ \lim_{n \to \infty} \hat{M}_n^{\ell}(a_{\ell+1} \mid a_1^{\ell}) = M(a_{\ell+1} \mid a_{\ell-k+1}^{\ell}). \]

By a similar evaluation in the proof of Lemma \ref{lem:SLLN:kth_Markov}, we see that
\begin{align*}
\frac{\log K_n}{n}
&\geq
 \sum_{a_1^{\ell} \in \Omega^{\ell}} \frac{T_n(a_1^{\ell})}{n} \cdot 
  D\left(\left. \hat{M}_n^{\ell}(\;\cdot \mid a_1^{\ell}) \right\| M(\;\cdot \mid a_{\ell-k+1}^{\ell}) \right) \\
&\ \ \  + \frac{1}{n} \left\{ -\log \tilde{P}(\omega_1^{\ell}) + \log \prod_{i= 1}^{\ell} M(\omega_i \mid \omega_{n-k+i}^n \omega_1^{i-1}) \right\}  -  \delta_{\ell} (n) \\
&\ \ \ +\frac{1}{n} \Big\{ c(\omega_1^n) \log c(\omega_1^n) - (-\log Q_{LZ}(\omega_1^n)) \Big\}.
\end{align*}
By the assumption of induction, $T_n(a_1^\ell)/n$ converges to a positive number $P(a_1^\ell)$ for all $a_1^\ell\in\Omega^\ell$ as $n\to\infty$. 
Therefore, if $\hat{M}_n^{\ell}(a_{\ell+1} \mid a_1^{\ell})$ does not converge to $M(a_{\ell+1} \mid a_{\ell-k+1}^{\ell})$ as $n\to\infty$,  we have $\limsup_{n \to \infty} K_n = \infty$. 
This completes the proof. 
\end{proof}

%==================================================================
\section{Applications}
%==================================================================

In this section, we provide some applications of Theorem \ref{thm:kth_Markov_normal_number}.

%-------------------------------------------------------------------------------------------------------------------
\subsection{Szil\'ard's engine game}
%-------------------------------------------------------------------------------------------------------------------

We begin by applying the framework of predictive games to thermodynamics, conceptualizing a thermodynamic cyclic as a betting game between Scientist and Nature%
\footnote{
A related perspective is discussed in \cite{Hiura}, which aims to give a game-theoretic characterization of Gibbs' distribution. Our approach differs by emphasizing the coding theoretic aspect of Szil\'ard's engine while adhering closely to the original formalism of the Shafer-Vovk theorem.
}.
As a fundamental prototype, we consider a work-extracting game inspired by Le\'o Szil\'ard's thought experiment \cite{Szilard}, which has been widely discussed in the context of the second law of thermodynamics and its relationship to information theory.

Consider the following work-extracting game played on a hypothetical engine illustrated in Figure~\ref{Figure-Szilard_game}:

\begin{figure}[t]
  \centering
  \includegraphics[width= 10cm]{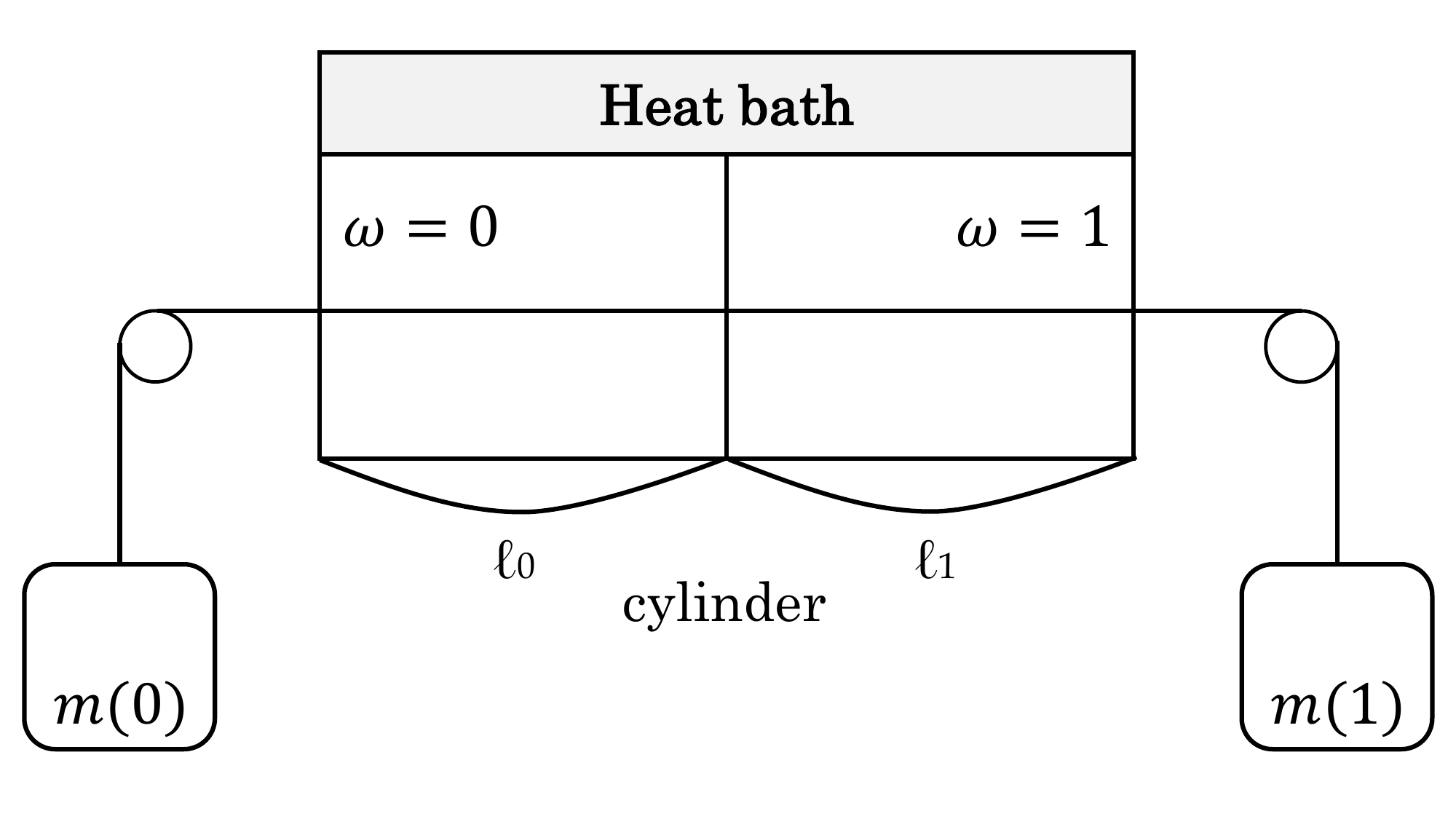}
  \caption{Szil\'ard's engine game.} \label{Figure-Szilard_game}
\end{figure}

\begin{enumerate}[(i)]
\item A partition, connected to two containers by inextensible strings, is placed at a specific position within a cylinder and fixed in place.
\item Scientist places a weight $m(0)$ on the left container and another weight $m(1)$ on the right container. 
\item Nature inserts a single molecule into one side of the partition, announces whether the molecule is in the left chamber ($\omega=0$) or the right chamber ($\omega=1$), and then releases the partition. 
\item If $\omega = 0$, the molecule pushes the partition to the right, and Scientist gains potential energy $m(0) g \ell_1 - m(1) g \ell_1$, where $g$ is the gravitational acceleration and $\ell_1$ is the displacement of the weights. 
If $\omega = 1$, on the other hand, Scientist instead gains potential energy $m(1) g \ell_0 - m(0) g \ell_0$.
\item Once the partition reaches the end of the cylinder, it is reset to its original position as in step (i). 
\end{enumerate}

Letting $\Omega := \{0,1\}$ and  
\[ r: = \frac{\ell_1}{\ell_0 + \ell_1}\in (0,1), \]
the above procedure can be formulated as a game-theoretic process:

\medskip
\begin{itembox}[l]{\bf Szil\'ard's engine game} \label{game:Szilard}
\textbf{Players} : Scientist and Nature. \\
\textbf{Protocol} : $W_0 = 1$. \\
\ \ FOR $n \in \mathbb{Z}_{>0}$ : \\
\ \ \ \ \ Scientist announces $m_n = (m_n(0), m_n(1)) \in \mathbb{R}^\Omega$. \\
\ \ \ \ \ Nature announces $\omega_n \in \Omega$. \\
\ \ \ \ \ $W_n := W_{n-1} + (m_n(1) - m_n(0)) g (\ell_0 + \ell_1) (\omega_n - r)$.\\
\ \ END FOR.
\end{itembox}

\medskip
At first glance, this game appears to favor Nature, as Nature announces $\omega_n$ after Scientist has set the weights. However, we can prove the following.

\begin{thm}\label{SLLN:Szilard}
In Szil\'ard's engine game, Scientist has a prudent strategy $\{m_n\}_n$ that ensures $\lim_{n \to \infty} W_n = \infty$ unless
\[
\lim_{n \to \infty} \frac{1}{n}\sum_{i=1}^n \delta_1 (\omega_i) = r.
\]
\end{thm}

\begin{proof}
Observe that the recurrence relation is rewritten as
\begin{equation*}\label{eqn:recursion}
 W_n := W_{n-1}\left[ 1+ \sum_{a\in\Omega} \alpha_n(a) \left( \delta_{\omega_n}(a)-p(a)\right) \right],
\end{equation*}
where $p:=(p(0),\,p(1)):=(1-r,\, r)$ and
\[
 \alpha_n(a):=\frac{m_n(a)g(\ell_0+\ell_1)}{W_{n-1}}.
\]
Thus, Theorem \ref{SLLN:Szilard} is an immediate consequence of Theorem \ref{SLLN:predictive-Shafer_and_Vovk}.
\end{proof}

The implication of Theorem \ref{SLLN:Szilard} is as follows:
If Nature does not behave in accordance with the expected statistical law, Scientist can extract an infinite amount of work from the engine. A distinctive feature of this finding is that it can be achieved without invoking Maxwell's demon \cite{LeffRex} or employing any measurement scheme to determine a molecule's position before setting weights. Instead, Scientist only needs to detect deviations in Nature's behavior from the statistical law using a universal data compression technique.
Note that this result bears a close resemblance to Kelvin's formulation of the second law of thermodynamics, which asserts that it is impossible to extract a net amount of work from a thermodynamic system while leaving the system in the same state.

\begin{figure}[t]
  \centering
  \includegraphics[width=10cm]{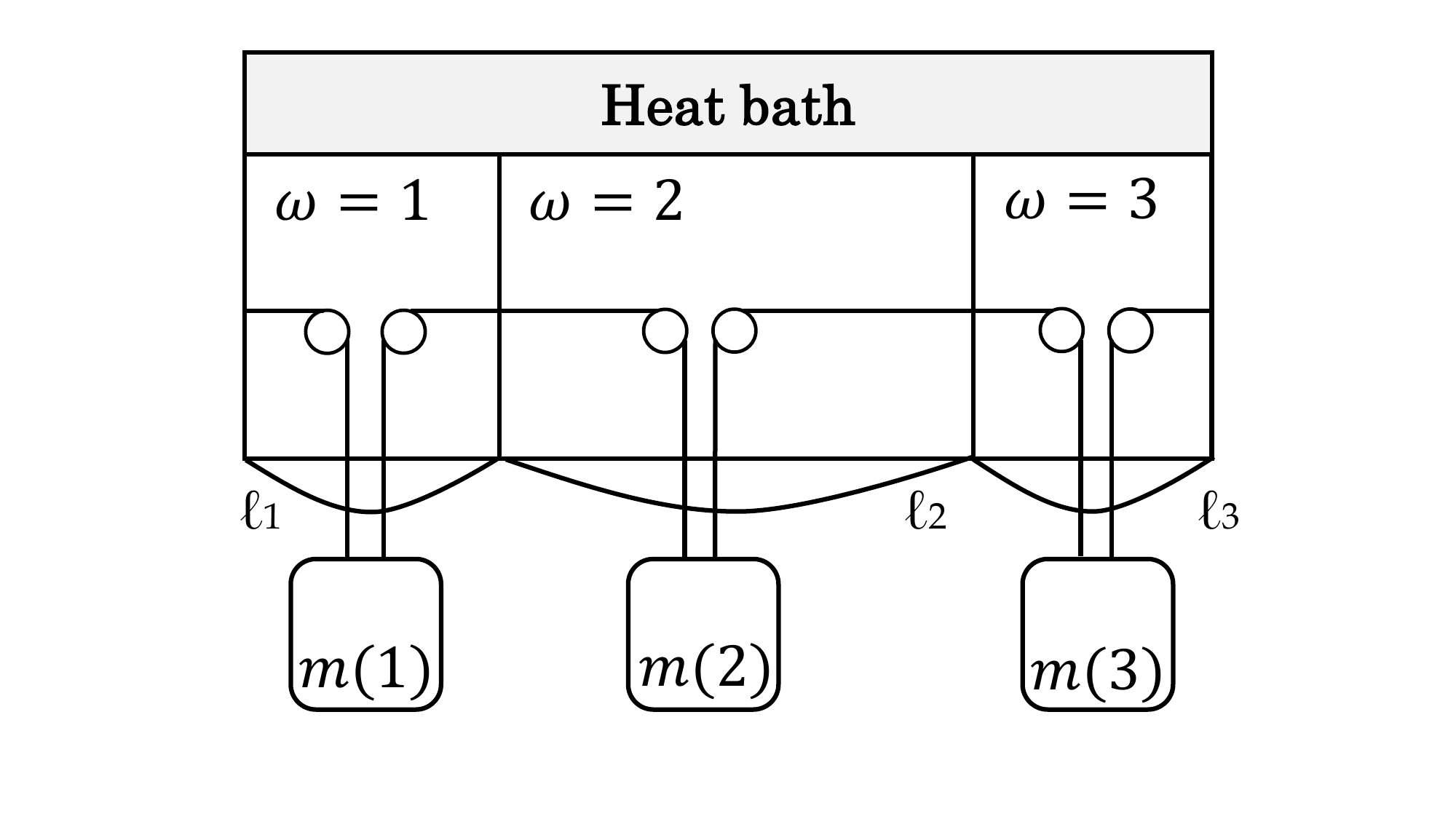}
  \caption{Generalized Szil\'ard's engine game having three chambers. The pulleys can move horizontally and are assumed to be negligibly small.}
  \label{Figure-Omega3}
\end{figure}

Extending the previous argument to the case when the outcome space $\Omega$ is an arbitrary finite set is straightforward.
Consider a device illustrated in Figure \ref{Figure-Omega3}, corresponding to the case when $\Omega =\{1,2,3\}$.
The cylinder contains two partitions, dividing it into three chambers labeled by $\omega= 1,2,3$.
Each partition is connected to two containers by inextensible strings and negligibly small pulleys that can move horizontally. Weights can be placed on these containers.
The containers correspond one-to-one with the chambers and are labeled accordingly.

A generalized Szilard's engine game for $\Omega =\{1,2, 3\}$ runs as follows: 
\begin{enumerate}[(i)]
\item Each of two partitions is placed at a specific position within the cylinder and fixed in place.
\item Scientist places a weight $m(a)$ on each container $a$ for $a = 1,2,3$.
\item Nature places a single molecule in one of the three chambers, announces its label $\omega$, and releases the partitions.
\item The molecule pushes the partitions at the boundaries of chamber $\omega$, causing the chamber to expand until all the partitions are pressed against the end(s) of the cylinder, and Scientist gains potential energy as follows:
If $\omega = 1$, the work extracted is
\[ m(1) g \frac{\ell_2+\ell_3}{2} - m(2) g \frac{\ell_2}{2} - m(3) g \frac{\ell_3}{2}. \]
If $\omega = 2$, the work extracted is
\[ m(2) g \frac{\ell_3+\ell_1}{2} - m(3) g \frac{\ell_3}{2} - m(1) g \frac{\ell_1}{2}. \] 
If $\omega = 3$, the work extracted is
\[ m(3) g \frac{\ell_1+\ell_2}{2} - m(1) g \frac{\ell_1}{2} - m(2) g \frac{\ell_2}{2}. \]
\item Once the partitions come to rest, they return to their original positions as in step (i). 
\end{enumerate}

\noindent
In a single round of the game, Scientist extracts the following amount of work:
\[
 \sum_{a= 1}^3 \frac{m(a) g}{2}  (\ell_1 +\ell_2 + \ell_3) \left(\delta_{\omega}(a) - \frac{\ell_a}{\ell_1 + \ell_2 + \ell_3} \right).
\]

Generalizing Szil\'ard's engine game to an arbitrary finite set $\Omega= \{1,2, \dots, A\}$ is straightforward: 
one simply increases the number of chambers illustrated in Figure~\ref{Figure-Omega3}.
Defining
\[ p(a):=\frac{\ell_a}{\ell_1 + \cdots + \ell_A}, \qquad (a\in\Omega), \]
we can formulate the generalized work-extracting protocol as follows.

\medskip
\begin{itembox}[l]{\bf Generalized Szil\'ard's engine game} \label{game:gene-Szilard}
\textbf{Players} : Scientist and Nature. \\
\textbf{Protocol} : $W_0 = 1$. \\
\ \ FOR $n \in \mathbb{Z}_{>0}$ : \\
\ \ \ \ \ Scientist announces $m_n \in \mathbb{R}^\Omega$. \\
\ \ \ \ \ Nature announces $\omega_n \in \Omega$. \\
\ \ \ \ \ $\displaystyle W_n := W_{n-1} + \sum_{a\in\Omega} \frac{m_n(a) g}{2} 
	(\ell_1 + \cdots + \ell_A) \left(\delta_{\omega_n}(a) - p(a)\right)$.\\
\ \ END FOR.
\end{itembox}

\medskip
Now, we extend Theorem \ref{SLLN:Szilard} to this generalized setting:

\begin{thm}\label{SLLN:gene-Szilard}
In generalized Szil\'ard's engine game, Scientist has a prudent strategy $\{m_n\}_n$ that ensures $\lim_{n \to \infty} W_n = \infty$ unless
\[
\lim_{n \to \infty} \frac{1}{n}\sum_{i=1}^n \delta_a (\omega_i) = p(a),\quad (\forall a\in\Omega).
\]
\end{thm}

We can further generalize the game described above by allowing the chamber size ratios $p(a)$ to vary in each round $n$, introducing a Forecaster who announces these ratios.
In this extended protocol, Theorem \ref{thm:kth_Markov_normal_number} provides a generalization of Theorem \ref{SLLN:gene-Szilard}, incorporating a time-homogeneous finite-order Markovian Forecaster.

%-------------------------------------------------------------------------------------------------------------------
\subsection{Entropy}
%-------------------------------------------------------------------------------------------------------------------

Given the pivotal role of universal coding schemes in establishing game-theoretic law of large numbers, it is natural to expect that the protocol of a predictive game is also intertwined with the concept of entropy.
The next proposition formalizes this connection, where we continue to assume that $M$ is a $k$th-order Markov kernel with strictly positive entries and $\pi$ is the stationary distribution of $M$.

\begin{prop}\label{prop:Markov_entropy}
In the time-homogeneous $k$th-order Markovian predictive game,
Skeptic has a prudent strategy $\alpha : \Omega^\ast \rightarrow \mathbb{R}^\Omega$ that ensures $\lim_{n \to \infty} K_n = \infty$ unless
\begin{equation}\label{eqn:compressionRate}
 \lim_{n \to \infty} \frac{-\log Q_{LZ}(\omega_1^n)}{n} = H(M),
\end{equation}
where 
\[
H(M):= -\sum_{a_1^k \in \Omega^k} \pi(a_1^k) \sum_{b \in \Omega} M(b \mid a_1^k) \log M(b \mid a_1^k)
\]
is the entropy rate. 
\end{prop}

\begin{proof}
We observed in the proof of Theorem \ref{thm:kth_Markov_normal_number} that, under the Lempel-Ziv betting strategy $Q_{LZ}$, the boundedness of the capital process, i.e., $\limsup_{n \to \infty} K_n < \infty$, guarantees not only that
\[
 \frac{\log K_n}{n}=\frac{1}{n} \left\{ -\log \tilde{P}(\omega_1^n) - ( -\log Q_{LZ}(\omega_1^n) ) \right\} \longrightarrow 0
\]
but also that
\[
\hat{M}_n^k (\;\cdot \mid a_1^k) \longrightarrow M(\;\cdot \mid a_1^k)
 \quad\mbox{and}\quad
 \frac{T_n(a_1^k)}{n} \longrightarrow \pi(a_1^k)
\]
for all $a_1^k\in\Omega^k$ as $n\to\infty$.

As a consequence, using a similar computation as in the proof of Lemma \ref{lem:lth_Markov_divergence}, we obtain
\begin{align*}
- \frac{1}{n}\log \tilde{P}(\omega_1^n)
&= -\sum_{a_1^k \in \Omega^k} \frac{T_n(a_1^k)}{n} \sum_{b \in \Omega} \hat{M}_n^k (b \mid a_1^k) \log M(b \mid a_1^k) \\
&\quad\;  +\frac{1}{n} \left\{-\log \tilde{P}(\omega_1^{k}) + \log \prod_{i= 1}^{k} M(\omega_i \mid \omega_{n-k+i}^n \omega_1^{i-1}) \right\} \\
& \longrightarrow H(M).
\end{align*}
Combining these asymptotic properties, \eqref{eqn:compressionRate} follows immediately. 
\end{proof}

The implication of Proposition \ref{prop:Markov_entropy} is as follows: To prevent Skeptic from becoming infinitely rich, Reality must ensure that the asymptotic compression rate of its moves coincides with the entropy rate.

Note that Proposition \ref{prop:Markov_entropy} bears a close resemblance to Lempel-Ziv's theorem \cite{LZ-78}
\[ \lim_{n \to \infty} \frac{\ell_{LZ}(\omega_1^n)}{n} = H_A (P),\quad \mbox{$P$-a.s.} \]
as well as Brudno's theorem \cite{Brudno}
\[ \lim_{n \to \infty} \frac{\mathcal{K}(\omega_1^n)}{n} = H_A (P),\quad \mbox{$P$-a.s.} \]
for the prefix Kolmogorov complexity $\mathcal{K}(\omega_1^n)$
when data $\omega_1^n$ are drawn according to a stationary ergodic probability measure $P$ on $\Omega^\infty$, where
\[ H_A (P):= \lim_{n \to \infty} \mathbb{E} \left[-\frac{1}{n} \log_A P(\omega_1^n) \right] \]
is the entropy rate to the base $A$. 
In addition, the last asymptotic property in the proof of Proposition \ref{prop:Markov_entropy} corresponds to the Shannon-McMillan-Breiman theorem \cite{CoverThomas}
\[
\lim_{n\to\infty} \left\{- \frac{1}{n} \log P(\omega_1^n)\right\} = H(P),\quad \mbox{$P$-a.s.}
\]
which also has a counterpart in algorithmic randomness theory \cite{{Vyugin:1998},{Nakamura},{Hochman}}. 
These observations prompt us to call a Forecaster {\it stationary ergodic} if they announce predictions according to the prescription 
\[ 
p_n(\omega_n) := P(\omega_n \mid \omega_1^{n-1}), 
\quad (\omega_1^n \in\Omega^n),
\]
where $P$ represents a predetermined stationary ergodic probability measure. 

If we were to discover a compression algorithm capable of efficiently compressing Reality's moves within the game-theoretic context, we could then define the entropy of a game as the asymptotic data compression rate, assuming that Reality faithfully follows the predictions of the stationary ergodic Forecaster, thereby preventing Skeptic from becoming infinitely rich. 

However, the direction of the above definition of the stationary ergodic Forecaster may not be fully satisfactory from the perspective of Dawid's prequential principle \cite{Dawid_prequential_1984}.
The validation and exploration of the concepts of game entropy and a stationary ergodic Forecaster remain topics for future investigation.

%==================================================================
\section{Concluding remarks}
%==================================================================

In this paper, we established a generalization of the game-theoretic law of large numbers in a time-homogeneous $k$th-order Markovian predictive game. By constructing a Lempel-Ziv-inspired strategy based on incremental parsing and the martingale properties of the game, we provided new insights into the relationship between game-theoretic randomness and coding theory.

We also explored applications to thermodynamics by formulating a game-theoretic version of Szilard's engine.
Our results demonstrated that Nature must behave stochastically, satisfying the law of large numbers, to avoid violating the second law of thermodynamics. Furthermore, we introduced the concept of entropy in predictive games, associating it to the codelength of universal coding.

Despite these advances, several important challenges remain. For instance, integrating additional thermodynamic concepts such as thermal equilibrium, thermal contact, temperature, and free energies into a game-theoretic framework remains a significant open problem. Additionally, extending the framework to non-Markovian processes could provide deeper insights into the dynamics of predictive games.

%==================================================================
\section*{Acknowledgements}
%==================================================================

We would like to express our gratitude to Professor Kenshi Miyabe for insightful discussions.
AF wishes to convey deep appreciation to the late Professor Tom Cover for his enlightening and thought-provoking discussions on Maxwell's demon.
The present study was supported by JSPS KAKENHI Grant Numbers 22340019, 17H02861, 23H01090, and 23K25787.

%==============================================================================
% Appendices ;----.----;----.----;----.----;----.----;----.----;----.----;----.----
%==============================================================================
\appendix
\section*{Appendix}
\setcounter{equation}{0}
\setcounter{thm}{0}
\setcounter{footnote}{1}
\addcontentsline{toc}{section}{Appendix}
\renewcommand{\thesubsection}{\Alph{subsection}}
\renewcommand{\thethm}{\Alph{subsection}.\arabic{thm}} %This changes the style of an already defined environment "thm" 
\renewcommand{\theequation}{\Alph{subsection}.\arabic{equation}}

%-------------------------------------------------------------------------------------------------------------------
\subsection{Proof of lemmas} \label{app:Lemmas}
%-------------------------------------------------------------------------------------------------------------------

In this appendix, we provide detailed proofs of the lemmas stated in Section \ref{subsec:lemma}.

\subsubsection{Proof of Lemma \ref{lem:liminf_comp_and_Q}}

Since
\begin{align*}
-\log Q_{LZ}(\omega_1^n) 
&= -\log Q_{LZ}(\omega_1^{n_T} \omega_{n_T +1}^n) 
= -\log Q_{LZ}(\omega_1^{n_T}) - \log Q_{LZ}(\omega_{n_T +1}^n \mid \omega_1^{n_T})
\end{align*}
for $n > n_T$, it follows from \eqref{eqn:evalQ_LZ} and \eqref{eqn:LZcodelength} that, for any $n$,
\begin{align*}
-\log Q_{LZ}(\omega_1^n) 
&\leq
 -\log Q_{LZ}(\omega_1^{n_T}) + \log |V(\omega_1^{n_T})|= 
\sum_{j=0}^{c(\omega_1^n)} \log (A +j(A-1)) \\
&< \sum_{j=0}^{c(\omega_1^n)} \log (A+c(\omega_1^n)(A-1))
= (c(\omega_1^n)+1) \log (A + c(\omega_1^n)(A-1)).
\end{align*}
Consequently, 
\begin{align*}
&\frac{1}{n}\left\{c(\omega_1^n)\log c(\omega_1^n) - (-\log Q_{LZ}(\omega_1^n)) \right\} \\
&\ \ > \frac{1}{n}\left\{ 
c(\omega_1^n)\log c(\omega_1^n) -(c(\omega_1^n)+1) \log (A + c(\omega_1^n)(A-1)) \right\} \\
&\ \ = \frac{c(\omega_1^n)}{n} \log \frac{c(\omega_1^n)}{ A + c(\omega_1^n)(A-1)} - \frac{\log (A + c(\omega_1^n)(A-1))}{n}.
\end{align*}
Thus, the following Lemma \ref{lem:comp_upper_bound} proves the claim.

\begin{lem}\label{lem:comp_upper_bound}
For sufficiently large $n$ and for all $\omega_1^n \in \Omega^n$,
\[
(0<) \ c(\omega_1^n) < \frac{n}{(1- \varepsilon_n) \log_A n},
\]
where $\varepsilon_n \rightarrow 0$ as $n \to \infty$.
Specifically, $c(\omega_1^n)/n \rightarrow 0$ as $n \to \infty$.
\end{lem}

\begin{proof}
See Lemma 13.5.3 of \cite{CoverThomas}.
\end{proof}

\subsubsection{Proof of Lemma \ref{lem:cond_empirical}}

The first identity follows from
\begin{align*}
\sum_{b \in \Omega} T_n(b \mid a_1^\ell)
&= \sum_{b \in \Omega} (\text{number of occurrences of $a_1^\ell b$ in $\omega_1^n \omega_1^\ell$}) \\
&= (\text{number of occurrences of $a_1^\ell$ in $\omega_1^n \omega_1^{\ell-1}$}) = T_n(a_1^\ell).
\end{align*}
On the other hand, observe that
\begin{align*}
\sum_{a_1 \in \Omega} T_n(b \mid a_1^\ell)
&= \sum_{a_1 \in \Omega} (\text{number of occurrences of $a_1^\ell b$ in $\omega_1^n \omega_1^\ell$}) \\
&= (\text{number of occurrences of $a_2^\ell b$ in $\omega_2^n \omega_1^{\ell}$}) \\
&= (\text{number of occurrences of $a_2^\ell b$ in $\omega_1^n \omega_1^{\ell-1}$}) + \Delta_1 + \Delta_\ell, 
\end{align*}
where
\begin{align*}
\Delta_1 
&:= (\text{adjustment for the effect of adding $\omega_1$ to the head of $\omega_2^n \omega_1^{\ell}$}) \\
&=\begin{cases}
-1 &\ \ (\omega_1^\ell = a_2^\ell b), \\
0 &\ \ (\text{otherwise}),
\end{cases}
\end{align*}
and
\begin{align*}
\Delta_\ell 
&:= (\text{adjustment for the effect of removing $\omega_\ell$ from the tail of $\omega_2^n \omega_1^{\ell}$}) \\
&=\begin{cases}
1 &\ \ (\omega_1^\ell = a_2^\ell b), \\
0 &\ \ (\text{otherwise}).
\end{cases}
\end{align*}
Since $\Delta_1 + \Delta_\ell = 0$, the second identity holds.

\subsubsection{Proof of Lemma \ref{lem:kth_Markov_stationary}}

The assumption $\hat{M}_n^k \to M$ ensures that for sufficiently large $n$,  we have $\hat{M}_n^k(b \mid a_1^k)>0$ for all $a_1^k \in \Omega^k$ and $b \in \Omega$.
For $\omega^k \in \Omega^k$, define $\hat{q}_n(\omega^k) := T_n(\omega^k)/n$.
Then, by Lemma \ref{lem:cond_empirical}, the empirical distribution $\hat{q}_n$ is stationary under the $k$th Markov kernel $\hat{M}_n^k$ \cite{Davisson-Longo-Sgarro}:
\[
\sum_{a_1 \in \Omega} \hat{M}_n^k (b \mid a_1^k) \hat{q}_n(a_1^k) 
= \sum_{a_1 \in \Omega} \frac{T_n(b \mid a_1^k)}{T_n(a_1^k)} \frac{T_n(a_1^k)}{n}
= \frac{1}{n} \sum_{a_1 \in \Omega} T_n(b \mid a_1^k) = \frac{T_n(a_2^k b)}{n} = \hat{q}_n(a_2^k b).
\]
Thus, by Lemma \ref{lem:kth_unique_stationary_distribution} in Appendix \ref{app:stationary}, $\hat{q}_n$ is the unique stationary distribution of Markov matrix $\hat{M}_n^k$.

Consider a convergent subsequence $\{\hat{q}_{n_i} -\pi \}_i$ of $\{\hat{q}_n -\pi\}_n$, which converges to some $r \in [-1,1]^{\Omega^k}$.
By the assumption $\hat{M}_n^k \rightarrow M$, we obtain
\begin{align*}
r &= \lim_{i \to \infty} (\hat{q}_{n_i}- \pi ) = \lim_{i \to \infty} ( \hat{M}_{n_i}^k \hat{q}_{n_i} - M \pi ) \\
&=\lim_{i \to \infty} \left\{ \hat{M}_{n_i}^k (\hat{q}_{n_i} - \pi) + (\hat{M}_{n_i}^k - M) \pi \right\} = M r.
\end{align*}
By the Perron-Frobenius theorem (Section 4.4 of \cite{Gallager}), 
there exists a constant $c \in \mathbb{R}$ satisfying $r = c \pi$.
Furthermore,
\[
\sum_{a_1^k \in \Omega^k} r(a_1^k) = \lim_{i \to \infty} \sum_{a_1^k \in \Omega^k} (\hat{q}_{n_i}(a_1^k) - \pi(a_1^k)) = 0.
\]
Thus, we must have $c = 0$, completing the proof.

\subsubsection{Proof of Lemma \ref{lem:Ziv_ineq}}

The proof is almost the same as that of Ziv's inequality and the asymptotical optimality of the Lempel-Ziv algorithms \cite[Section 13.5.2]{CoverThomas}.
Suppose the sequence $\omega_1^n$ is parsed into $C_n$ distinct substrings as 
\[ \omega_1^n = /\omega_{n_0 +1}^{n_1} /\omega_{n_1 +1}^{n_2} /\cdots / \omega_{n_{C_n-1} +1}^{n_{C_n}}, \]
where $n_0 =0$ and $n_{C_n}= n$.
For example, if $\omega_1^n$ is parsed using the incremental parsing algorithm as
\[ 
 \omega_1^n = /\omega_{n_0 + 1}^{n_1} / \omega_{n_1 + 1}^{n_2} / \cdots /\omega_{n_{T-1} +1}^{n_T} /\omega_{n_T +1}^{n},
\]
we define $C_n = T$ and set $\omega_{n_{C_n-1} +1}^{n_{C_n}} :=  \omega_{n_{T-1} +1}^{n_T} \omega_{n_T +1}^{n}$.

Now, define $s_i := \omega_{i-\ell}^{i-1}$ for $\ell+1 \leq i \leq n$, and extend them cyclically for $1\le i \le\ell$ as follows:
\[
 s_1 := \omega_{n-\ell+1}^n,\;\; s_2 := \omega_{n-\ell+2}^n \omega_1,\;\; \dots, \;\; s_\ell := \omega_n \omega_1^{\ell-1}.
\]
For $m \in \mathbb{Z}_{>0}$ and $s \in \Omega^\ell$, let $c_{m,s}$ denote the number of occurrences of the word $\omega_{n_{j-1} +1}^{n_j}$ of length $m$ such that $s_{n_{j-1}+1} = \omega_{n_{j-1}-\ell+1}^{n_{j-1}} =s$ among the $C_n$ substrings $\omega_{n_0 +1}^{n_1}, \omega_{n_1 +1}^{n_2}, \dots, \omega_{n_{C_n -1} +1}^{n_{C_n}}$, i.e.,
\[
c_{m,s} := \left| \left\{ j \in \{1,2,\dots, C_n\} \left| \;|\omega_{n_{j-1} +1}^{n_j}| = m,\; s_{n_{j-1} +1} = s \right. \right\}\right|.
\]
Letting $\mathcal{J}:=\{ (m,s)\in \mathbb{Z}_{>0}\times\Omega^\ell \mid  c_{m,s} > 0\}$,
we have 
\[
\sum_{(m,s)\in \mathcal{J}} c_{m,s} = C_n \;\;\text{ and }\;\;
\sum_{(m,s)\in \mathcal{J}} m \cdot c_{m,s} = n.
\]
With a slight abuse of notation, we define, for any $m \in \mathbb{Z}_{>\ell}$ and $a_1^m \in \Omega^m$,
\[
\hat{M}_n^\ell (a_{\ell+1}^m \mid a_1^\ell) :=\prod_{i = \ell+1}^m \hat{M}_n^\ell (a_i \mid a_{i-\ell}^{i-1}).
\]
Then, we can evaluate $\log \hat{R}_n^\ell(\omega_1^n)$ as follows:
\begin{align*}
\log \hat{R}_n^\ell(\omega_1^n)
&=\sum_{j=1}^{C_n} \log \hat{M}_n^\ell(\omega_{n_{j-1} +1}^{n_j} \mid s_{n_{j-1} +1}) \\
&=\sum_{(m,s)\in \mathcal{J}} \; \sum_{\substack{j : n_j - n_{j-1} = m, \\ s_{n_{j-1} +1}= s }}\log \hat{M}_n^\ell(\omega_{n_{j-1} +1}^{n_j} \mid s_{n_{j-1} +1}) \\
&=\sum_{(m,s)\in \mathcal{J}} c_{m,s} \; \sum_{\substack{j : n_j - n_{j-1} = m, \\ s_{n_{j-1} +1}= s }} \frac{1}{c_{m,s}} \log \hat{M}_n^\ell(\omega_{n_{j-1} +1}^{n_j} \mid s_{n_{j-1} +1}) \\
&\leq \sum_{(m,s)\in \mathcal{J}} c_{m,s} \log \left( \frac{1}{c_{m,s}} \sum_{\substack{j : n_j - n_{j-1} = m, \\ s_{n_{j-1} +1}= s }} \hat{M}_n^\ell(\omega_{n_{j-1} +1}^{n_j} \mid s_{n_{j-1} +1}) \right).
\end{align*}
In the last inequality, we used Jensen's inequality.
Since the parsed substrings $\{\omega_{n_{j-1} +1}^{n_j}\}_{1\le j \le C_n}$ are distinct, we have
\[
\sum_{\substack{j : n_j - n_{j-1} = m, \\ s_{n_{j-1} +1}= s }} \hat{M}_n^\ell(\omega_{n_{j-1} +1}^{n_j} \mid s_{n_{j-1} +1}) \leq 1
\]
for all $(m,s)\in\mathcal{J}$. 
As a consequence,
\begin{align*}
\log \hat{R}_n^\ell(\omega_1^n)
&\leq - \sum_{(m,s)\in \mathcal{J}} c_{m,s} \log c_{m,s} \\
&= -c(\omega_1^n) \log c(\omega_1^n) - c(\omega_1^n) \sum_{(m,s)\in \mathcal{J}} \frac{c_{m,s}}{c(\omega_1^n)} \log \frac{c_{m,s}}{c(\omega_1^n)},
\end{align*}
where $c(\omega_1^n)=C_n$.
Writing $\pi_{m,s} := c_{m,s}/c(\omega_1^n)$, we have
\[
\sum_{(m,s)\in \mathcal{J}}  \pi_{m,s} = 1 \;\; \text{ and }\;\;  \sum_{(m,s)\in \mathcal{J}}  m\cdot \pi_{m,s} = \frac{n}{c(\omega_1^n)}. 
\]

We now define the random variables $U$ and $V$ as follows:
\[
\mathrm{Pr}(U = m,\ V = s) := \pi_{m,s}.
\]
From the above bound on $\log \hat{R}_n^\ell(\omega_1^n)$, it follows that
\[
-\frac{1}{n} \log \hat{R}_n^\ell(\omega_1^n)
\geq \frac{c(\omega_1^n)}{n} \log c(\omega_1^n) - \frac{c(\omega_1^n)}{n} H(U,V),
\]
where
\[
H(U,V) := - \sum_{(m,s)\in \mathcal{J}} \; \frac{c_{m,s}}{c(\omega_1^n)} \log \frac{c_{m,s}}{c(\omega_1^n)}.
\]
By the subadditivity of entropy, we have
\[
H(U,V) \leq H(U) + H(V).
\]
Since the expectation of $U$ is given by
\[ \mathbb{E}[U] = \frac{n}{c(\omega_1^n)}, \]
applying Lemma \ref{lem:integer-valued} below, we can bound $H(U)$ as
\begin{align*}
H(U)
&\leq (\mathbb{E}[U] +1)\log(\mathbb{E}[U] +1) - (\mathbb{E}[U])\log(\mathbb{E}[U]) \\
&= \log \frac{n}{c(\omega_1^n)} + \left( \frac{n}{c(\omega_1^n)} +1 \right)\log\left( \frac{c(\omega_1^n)}{n} +1 \right).
\end{align*}
On the other hand, since $H(V) \leq \log |\Omega|^\ell = \ell \log A$, we obtain
\begin{align*}
\delta_{\ell}(n) 
&:=
\frac{c(\omega_1^n)}{n} H(U,V) \\
&\leq \frac{c(\omega_1^n)}{n} \log \frac{n}{c(\omega_1^n)} 
 + \left( 1 +\frac{c(\omega_1^n)}{n} \right)\log\left( \frac{c(\omega_1^n)}{n} +1 \right) + \frac{c(\omega_1^n)}{n} \ell \log A .
\end{align*}
Since $c(\omega_1^n)/n \rightarrow 0$ as $n \to \infty$ by Lemma \ref{lem:comp_upper_bound}, it follows that $\delta_\ell(n) \rightarrow 0$ as $n \to \infty$.
This completes the proof.

\begin{lem}\label{lem:integer-valued}
Let $Z$ be a nonnegative integer-valued random variable with mean $\mu$. Then the entropy $H(Z)$ is bounded by
\[
H(Z) \leq (\mu + 1)\log(\mu +1) - \mu \log \mu.
\]
\end{lem}

\begin{proof}
See Lemma 13.5.4 of \cite{CoverThomas}.
\end{proof}

\subsubsection{Proof of Lemma \ref{lem:lth_Markov_divergence}}

For $n, \ell\in\mathbb{Z}_{>0}$ satisfying $n > \ell \geq k$,
\begin{align*}
&-\log \tilde{P}(\omega_1^n) + \log \hat{R}^{\ell}_n (\omega_1^n) \\
&\quad = -\log \tilde{P}(\omega_1^\ell) -\log \tilde{P}(\omega_{\ell+1}^n \mid \omega_1^\ell) + \log \hat{R}^{\ell}_n (\omega_1^n) \\
&\quad = -\log \tilde{P}(\omega_1^{\ell}) - \log \prod_{i= \ell+1}^n M(\omega_i \mid \omega_{i-k}^{i-1}) + \log \hat{R}^{\ell}_n (\omega_1^n) \\
&\quad =  -\log \tilde{P}(\omega_1^{\ell}) + \log \prod_{i= 1}^{\ell} M(\omega_i \mid \omega_{n-k+i}^n \omega_1^{i-1}) \\
& \qquad -\log \prod_{i= 1}^{\ell} M(\omega_i \mid \omega_{n-k+i}^n \omega_1^{i-1}) -\log \prod_{i= \ell+1}^n M(\omega_i \mid \omega_{i-k}^{i-1}) + \log \hat{R}^{\ell}_n (\omega_1^n) \\
&\quad =  -\log \tilde{P}(\omega_1^{\ell}) + \log \prod_{i= 1}^{\ell} M(\omega_i \mid \omega_{n-k+i}^n \omega_1^{i-1}) \\
& \qquad - \sum_{a_1^{\ell} \in \Omega^{\ell}} \sum_{b \in \Omega} T_n(b \mid a_1^{\ell}) \log M(b \mid a_{\ell-k+1}^{\ell})  + \log \hat{R}^{\ell}_n (\omega_1^n).
\end{align*}
In the last equality, we used the fact that 
\[ 
M(\omega_i \mid \omega_{n-k+i}^n \omega_1^{i-1})=M(\omega_i \mid \omega_{n-\ell+i}^n \omega_1^{i-1})
 \;\;\mbox{and}\;\;
 M(\omega_i \mid \omega_{i-k}^{i-1})=M(\omega_i \mid \omega_{i-\ell}^{i-1}),
\]
since $\ell \geq k$ and $M$ is the $k$th-order Markov kernel.

Substituting the definition of $\hat{R}^{l}_n (\omega_1^n)$, the computation follows as:
\begin{align*}
&-\log \tilde{P}(\omega_1^n) + \log \hat{R}^{\ell}_n (\omega_1^n) \\
&\quad = -\log \tilde{P}(\omega_1^{\ell}) + \log \prod_{i= 1}^{\ell} M(\omega_i \mid \omega_{n-k+i}^n \omega_1^{i-1}) \\ 
&\qquad + \sum_{a_1^{\ell} \in \Omega^{\ell}} \sum_{b \in \Omega} T_n(b \mid a_1^\ell) \left(-\log M(b \mid a_{\ell-k+1}^{\ell}) + \log \hat{M}_n^\ell(b \mid a_1^\ell) \right) \\
&\quad = -\log \tilde{P}(\omega_1^{\ell}) + \log \prod_{i= 1}^{\ell} M(\omega_i \mid \omega_{n-k+i}^n \omega_1^{i-1}) \\
&\qquad + \sum_{a_1^{\ell} \in \Omega^{\ell}} T_n(a_1^{\ell}) \cdot D\left(\left. \hat{M}_n^{\ell}(\;\cdot \mid a_1^{\ell-1}) \right\| M(\;\cdot \mid a_{\ell-k+1}^{\ell}) \right).
\end{align*}
This completes the proof.

%-------------------------------------------------------------------------------------------------------------------
\subsection{Stationary distributions of Markov chains} \label{app:stationary}
%-------------------------------------------------------------------------------------------------------------------

Given a Markov matrix $M : \Omega\times \Omega \to (0,1): (a, b) \mapsto M(a \mid b)$ satisfying
\[ \sum_{a \in \Omega} M(a \mid b) = 1  \] 
for all $b \in \Omega$,
let $M^{(m)}$ be the $m$th power of $M$, in that $M^{(1)} := M$ and 
\[
M^{(m)}(a\mid b) := \sum_{c\in \Omega} M(a\mid c) M^{(m-1)}(c \mid b).
\]
We recall the following well-known fact.

\begin{lem}\label{lem:mixed_finite_Markov_chain_kth}
There exists a unique probability distribution $\mu$ on $\Omega$ such that for any $a, b \in \Omega$,
\[
\lim_{m \to \infty}M^{(m)}(a \mid b) = \mu(a),
\]
and $\mu$ is the stationary distribution of $M$.
\end{lem}

\begin{proof}
See Theorem 6 in Chapter 4 of \cite{Gallager}.
\end{proof}

Now, consider a $k$th-order Markov kernel $M : \Omega\times \Omega^k \to (0,1): (a, \omega_1^k) \mapsto M(a \mid \omega_	1^k)$ that satisfies
\[ \sum_{a \in \Omega} M(a \mid \omega_1^k) = 1 \qquad ( \forall \omega_1^k \in \Omega^k).  \]

\begin{lem}\label{lem:kth_unique_stationary_distribution}
For the $k$th-order Markov kernel $M$, there is a unique stationary distribution $\pi : \Omega^k \to (0,1)$ satisfying
\[
\pi(a_2^{k+1}) = \sum_{a_1\in \Omega} M(a_{k+1} \mid a_1^k) \pi(a_1^k).
\]
\end{lem}

\begin{proof}
Define $\tilde{M}: \Omega^k \times \Omega^k \to [0,1)$ by
\begin{align*}
\tilde{M}(a_1^k \mid b_1^k) := 
\begin{cases}
M(a_k \mid b_1^k) &\ \ (a_1^k = b_2^k a_k), \\
0 &\ \ (\text{otherwise}).
\end{cases}
\end{align*}
Since
\[
\sum_{a_1^k \in \Omega^k} \tilde{M}(a_1^k \mid b_1^k) = \sum_{a_k \in \Omega} \sum_{a_1^{k-1} \in \Omega^{k-1}} \tilde{M}(a_1^k \mid b_1^k) = \sum_{a_k \in \Omega} M(a_k \mid b_1^k) = 1,
\]
we can regard $\tilde{M}$ as a first-order Markov kernel on $\Omega^k$.
Moreover it is straightforward to verify that
\[
 \tilde{M}^{(k)} (a_1^k \mid b_1^k) = \prod_{i=1}^k M(a_i \mid b_i^k a_1^{i-1}).
\]
This expression ensures that $\tilde{M}^{(k)} (a_1^k \mid b_1^k)>0$ for all $a_1^k, b_1^k \in \Omega^k$. 
Thus, by applying Lemma \ref{lem:mixed_finite_Markov_chain_kth} to the Markov matrix $\tilde M^{(k)}$, 
we conclude that there exists a unique distribution $\pi$ on $\Omega^k$ satisfying
\[
\tilde{M}^{(k)} \pi = \pi.
\]
Furthermore, since
\[
\tilde{M}^{(k)} (M \pi) = M (\tilde{M}^{(k)} \pi) = M \pi,
\]
the uniqueness of the stationary distribution for $\tilde{M}^{(k)}$ implies that $M \pi = \pi$.
\end{proof}

%-------------------------------------------------------------------------------------------------------------------
\subsection{Lynch-Davisson betting strategy for simple predictive game} \label{app:Lynch-Davisson}
%-------------------------------------------------------------------------------------------------------------------

In this appendix, we present an alternative proof of Theorem \ref{SLLN:predictive-Shafer_and_Vovk} using one of the simplest universal data compression schemes \cite{Lynch, Davisson}.
As a by-product, we also analyze the convergence rate of the empirical distribution.

We begin with a binary case and consider describing a binary sequence $x^n=11001$ of length $n=5$.
For $a\in\{0,1\}$, let $S_n(a)$ denote the number of occurrences of $a$ in $x^n$.
The sequence can be identified by first specifying its type (also known as the empirical distribution): 
\[
  \hat P_{x^n}=\left(\frac{S_n(0)}{n}, \frac{S_n(1)}{n}\right)=\left(\frac{2}{5}, \frac{3}{5}\right)
\]
and then specifying the index of this sequence among all sequences of length $n=5$ that share this type. 
Thus, the given binary sequence $x^n$ can be described by another binary sequence as follows: 
\[
 \underbrace{\fbox{ specify the type }}_{\lceil \log_2(n+1) \rceil \;\;\mbox{bits}} 
 \;+\;
 \underbrace{\fbox{ specify the sequence in the type class }}_{\left\lceil \log_2  {\small \begin{pmatrix} n \\ S_n(0), \, S_n(1) \end{pmatrix}} \right\rceil \;\;\mbox{bits} }
\]
This scheme is called the {\it Lynch-Davisson code}, and its codelength is given by
\[
 \ell_{LD}(x^n)=\lceil \log_2(n+1) \rceil + \left\lceil \log_2  \frac{n!}{S_n(0)! \, S_n(1)!}\right\rceil
 =\log_2  \frac{(n+1)!}{S_n(0)! \, S_n(1)!}+O(1). 
\]

Generalizing to a generic alphabet $\Omega=\{1,2,\dots, A\}$ is straightforward, and the corresponding Lynch-Davisson codelength is
\begin{eqnarray*}
 \ell_{LD}(x^n)
 &=&\left\lceil \log_A\frac{(n+A-1)!}{n!(A-1)!} \right\rceil 
 + \left\lceil \log_A\frac{n!}{S_n(1)! \, S_n(2)! \cdots \, S_n(A)!} \right\rceil \\
 &=&\log_A\frac{(n+A-1)!}{(A-1)!\,S_n(1)! \, S_n(2)! \cdots \, S_n(A)!} +O(1).
\end{eqnarray*}
Now, we are ready to prove Theorem \ref{SLLN:predictive-Shafer_and_Vovk}.

\begin{proof}[Proof of Theorem \ref{SLLN:predictive-Shafer_and_Vovk}]
Let us introduce the reference probability measure $P$ on $\Omega^*$ defined by
\[ P(\omega_1^n):=\prod_{i=1}^n p(\omega_i), \]
and consider the ``randomness deficiency'' function $\mathcal{L}_{LD}(\omega_1^n)$ for the Lynch-Davisson codelength $\ell_{LD}(\omega_1^n)$ relative to the Shannon codelength $ -\log_A P(\omega_1^n)$ defined by 
\begin{equation*}
\mathcal{L}_{LD}(\omega_1^n) := -\log_A P(\omega_1^n) - \ell_{LD}(\omega_1^n).
\end{equation*}
A crucial observation is that
\begin{align}
\mathcal{L}_{LD}(\omega_1^n)
&=-\log_A \prod_{i=1}^n p(\omega_i)+\log_A \frac{(A-1)!\, S_n(1)!\, S_n(2)! \cdots \, S_n(A)!}{(n+A-1)!}+O(1) \nonumber \\
&= -\sum_{a\in\Omega} S_n(a) \log_A p(a)
+ \log_A \frac{(A-1)!\, S_n(1)!\, S_n(2)!\, \cdots \, S_n(A)!}{(n+A-1)!}+O(1)  \label{eqn:L0} \\
&= (\log_A e)\left\{ \sum_{a\in\Omega} S_n(a) \log \frac{S_n(a)}{p(a)} - n \log n - O(\log n) \right\}
 \nonumber \\
&= n (\log_A e) \left\{ D( \hat{P}_{\omega_1^n}  \| p ) - O\left(\frac{\log n}{n}\right) \right\},
\label{eqn:L}
\end{align}
where Stirling's formula was used in the third equality.

The relation \eqref{eqn:L} shows that $\limsup_{n\to\infty}\mathcal{L}_{LD}(\omega_1^n)=\infty$ if 
$\hat{P}_{\omega_1^n}$ does not converge to $p$.
It then suffices to show that there exists a prudent betting strategy $\alpha_n$ that realizes
\begin{equation}\label{eqn:Wn_binary}
K_n \propto A^{\mathcal{L}_{LD}(\omega_1^n)} 
= \frac{(A-1)!\,S_n(1)! \, S_n(2)! \cdots \, S_n(A)!}{\prod_{i = 1}^n p(\omega_i) \cdot (n+A-1)!}.
\end{equation}
If this were the case, then
\begin{align}
\frac{K_n}{K_{n-1}}
& = \frac{1}{p(\omega_n) (n+A-1)}\cdot \frac{S_n(1)!\, S_n(2)! \cdots \, S_n(A)!}{S_{n-1}(1)!\, S_{n-1}(2)! \cdots S_{n-1}(A)!} \nonumber \\
& = \frac{1}{p(\omega_n) (n+A-1)}\cdot (S_{n-1}(\omega_n) +1). 
\label{eqn:ratio_binary}
\end{align}
Comparing this with the recursion formula
\[
 K_n := K_{n-1}\left\{ 1+ \sum_{a\in\Omega} \alpha_n(a) \left( \delta_{\omega_n}(a)-p(a)\right) \right\},
\]
we find that
\begin{equation}\label{eqn:alpha_binary}
\alpha_n(a) := \frac{S_{n-1}(a) + 1}{p(a)(n + A-1)}
\end{equation}
gives a desired prudent betting strategy that satisfies \eqref{eqn:Wn_binary}. In fact, since
\[
 \sum_{a\in\Omega}\alpha_n (a) p(a)=
 \frac{1}{n+A-1} \sum_{a\in\Omega} \left\{ S_{n-1}(a)+1 \right\}
 =\frac{1}{n+A-1} \left\{ (n-1)+A \right\}
 =1,
\]
we have
\[
 1 +  \sum_{a \in \Omega} \alpha_n (a) (\delta_{\omega_n} (a) - p(a))
 = \sum_{a \in \Omega} \alpha_n (a) \delta_{\omega_n} (a) 
 = \alpha_n (\omega_n),
\]
which is identical to the right-hand side of \eqref{eqn:ratio_binary}. 

In summary, the prudent betting strategy \eqref{eqn:alpha_binary} ensures that 
\begin{equation*}
 \limsup_{n\to\infty} \log_A K_n = \limsup_{n\to\infty} \mathcal{L}_{LD}(\omega_1^n) =\infty
\end{equation*}
if $\hat P_{\omega_1^n}$ does not converge to $p$ as $n \to \infty$.
The proof is complete.
\end{proof}

\begin{rem}\label{rem:speed_of_convergence}
In Theorem \ref{SLLN:predictive-Shafer_and_Vovk}, the two events $K_n\to\infty$ and $\hat P_{\omega_1^n}\to p$ are not necessarily mutually exclusive, and both may occur simultaneously%
\footnote{
A similar argument is found in \cite{Kumon-Takemura-Takeuchi}.}.
For example, suppose that $\hat P_{\omega_1^n}$ converges to $p$ at the rate
\[ \|\hat P_{\omega_1^n}-p \|=O(\sqrt{{\log n}/{n}}\,) \]
and satisfies
\begin{equation*}
  \limsup_{n\to\infty} \frac{n}{\log n}\sum_{a\in\Omega} \frac{(\hat P_{\omega_1^n}(a)-p(a))^2}{p(a)}>A-1.
\end{equation*}
Then, we have
\[  \limsup_{n\to\infty} K_n=\infty. \]
\end{rem}

\begin{proof}
Applying Stirling's formula
\[
\log n! = \left(n + \frac{1}{2} \right)\log n - n + O(1), 
\]
we get
\begin{align*}
& \log\frac{S_n(1)!\, S_n(2)! \cdots \, S_n(A)!}{(n+A-1)!} \\
&\qquad =  \sum_{a \in \Omega} \left\{ \left( S_n(a) + \frac{1}{2} \right)\log S_n(a) - S_n(a)  \right\}\\
&\qquad\qquad - \left\{ \left( (n +A- 1) + \frac{1}{2} \right) \log(n +A- 1) - (n +A- 1) \right\}  + O(1)\\
&\qquad= \sum_{a \in \Omega} S_n(a)\log S_n(a) + \frac{1}{2}\sum_{a \in \Omega} \log S_n(a) 
 - \left(n+A-\frac{1}{2}\right) \log (n+1) + O(1).
\end{align*}
Thus, the randomness deficiency function $\mathcal{L}_{LD}(\omega_1^n)$ is evaluated using \eqref{eqn:L0} as
\begin{align}
\frac{\mathcal{L}_{LD}(\omega_1^n)}{n \log_A e}
&= - \sum_{a \in \Omega} \frac{S_n(a)}{n} \log p(a)
+ \frac{1}{n} \log \frac{S_n(1)!\, S_n(2)! \cdots \, S_n(A)!}{(n+A-1)!} + O\left(\frac{1}{n} \right) \nonumber \\
&=\sum_{a\in\Omega} \frac{S_n(a)}{n} (-\log p(a) +  \log S_n(a)) 
	+\frac{1}{2n}\sum_{a \in \Omega} \log S_n(a)  \nonumber \\
&\qquad  -\frac{1}{n}\left(n+A-\frac{1}{2}\right) \log(n+A-1)+ O\left(\frac{1}{n} \right) \nonumber \\
&= \sum_{a \in \Omega} \hat P_{\omega_1^n}(a) \left( \log \frac{\hat P_{\omega_1^n}(a)}{p(a)} + \log n \right) 
+ \frac{1}{2n}\sum_{a \in \Omega}  \left( \log \hat P_{\omega_1^n}(a) + \log n \right) \nonumber \\
&\qquad  -\left(1+\frac{A}{n}-\frac{1}{2n}\right) \log(n+A-1)+ O\left(\frac{1}{n} \right) \nonumber \\
&= \sum_{a \in \Omega} \hat P_{\omega_1^n}(a) \log \frac{\hat P_{\omega_1^n}(a)}{p(a)} 
+ \frac{1}{2n}\sum_{a \in \Omega} \log \hat P_{\omega_1^n}(a) - \frac{A-1}{2n} \log n + O\left(\frac{1}{n} \right). \qquad \ \ 
 \label{eqn:used_Stirling} 
\end{align} 
Letting $Q_n(a) := \hat P_{\omega_1^n}(a) -p(a)$, we evaluate the first term of \eqref{eqn:used_Stirling} as
\begin{align}
\sum_{a \in \Omega} \hat P_{\omega_1^n}(a) \log \frac{\hat P_{\omega_1^n}(a)}{p(a)}
&= \sum_{a \in \Omega} p(a) \left( 1 + \frac{Q_n(a)}{p(a)} \right) \log \left( 1 + \frac{Q_n(a)}{p(a)} \right) \nonumber \\
&= \sum_{a \in \Omega} p(a) \left\{ \frac{Q_n(a)}{p(a)} + \frac{1}{2} \left(\frac{Q_n(a)}{p(a)}\right)^2 
	- O\left(\frac{Q_n(a)}{p(a)}\right)^3 \right\} \nonumber \\
&= \label{eqn:used_Taylor} 0 + \frac{1}{2} \sum_{a \in \Omega} \frac{Q_n(a)^2}{p(a)} 
	+ O\left(|Q_n|^3 \right).
\end{align}
Combining \eqref{eqn:used_Stirling} and \eqref{eqn:used_Taylor}, we have
\begin{equation}
\frac{\mathcal{L}_{LD}(\omega_1^n)}{n \log_A e}
= \frac{1}{2} \sum_{a \in \Omega} \frac{Q_n(a)^2}{p(a)} -\frac{A-1}{2n} \log n + \frac{1}{2n} \sum_{a \in \Omega} \log \hat P_{\omega_1^n}(a) 
+O\left(|Q_n|^3\right) + O\left(\frac{1}{n} \right),  \nonumber
\end{equation}
and thus
\begin{align}
\frac{\mathcal{L}_{LD}(\omega_1^n)}{\log_A e}
&= \frac{\log n}{2} \left[ \frac{n}{\log n} \sum_{a \in \Omega} \frac{Q_n(a)^2}{p(a)} -(A-1) \right]+ \sum_{a \in \Omega} \log \hat P_{\omega_1^n}(a)  \nonumber \\
&\qquad +n O\left(|Q_n|^3\right) + O(1).
\label{eqn:asymptoticL}
\end{align}
Now, by the assumption that $|Q_n|=|\hat P_{\omega_1^n}-p |=O(\sqrt{{\log n}/{n}}\,)$, we have
\[
 \sum_{a \in \Omega} \log\hat P_{\omega_1^n}(a) = O(1)
 \quad\mbox{and}\quad 
 n O\left(|Q_n|^3\right)\to 0.
\]
It then follows from \eqref{eqn:asymptoticL} that
\[
\limsup_{n\to\infty} \frac{n}{\log n}\sum_{a\in\Omega} \frac{(\hat P_{\omega_1^n}(a)-p(a))^2}{p(a)}>A-1, 
\]
implies $\limsup_{n\to\infty} \mathcal{L}_{LD}(\omega_1^n)=\infty$. 
This completes the proof. 
\end{proof}

Note that the quantity
\[ \sum_{a\in\Omega} \frac{(\hat P_{\omega_1^n}(a)-p(a))^2}{p(a)} \]
corresponds to the Fisher information.

%==================================================================

\end{document}